\documentclass[amsart,12pt]{article}
\usepackage{bm}
\usepackage{tikz}
\usepackage{tikz,pgfplots}
\usepackage{caption}
\usepackage{color}
\usepackage{amsmath}
\usepackage{amssymb}\usepackage{amsthm}
\usepackage{amsfonts}
\usepackage{amscd}
\begin{document}
\theoremstyle{plain}
\newtheorem*{ithm}{Theorem}
\newtheorem*{iprop}{Proposition}
\newtheorem*{idefn}{Definition}
\newtheorem{thm}{Theorem}[section]
\newtheorem{lem}[thm]{Lemma}
\newtheorem{dlem}[thm]{Lemma/Definition}
\newtheorem{prop}[thm]{Proposition}
\newtheorem{set}[thm]{Setting}
\newtheorem{cor}[thm]{Corollary}
\newtheorem*{icor}{Corollary}
\theoremstyle{definition}
\newtheorem{assum}[thm]{Assumption}
\newtheorem{notation}[thm]{Notation}
\newtheorem{setting}[thm]{Setting}
\newtheorem{defn}[thm]{Definition}
\newtheorem{clm}[thm]{Claim}
\newtheorem{ex}[thm]{Example}
\theoremstyle{remark}
\newtheorem{rem}[thm]{Remark}
\newcommand{\unit}{\mathbb I}
\newcommand{\ali}[1]{{\caB}_{[ #1 ,\infty)}}
\newcommand{\alm}[1]{{\caB}_{(-\infty, #1 ]}}
\newcommand{\nn}[1]{\lV #1 \rV}
\newcommand{\br}{{\mathbb R}}
\newcommand{\dm}{{\rm dom}\mu}
\newcommand{\lb}{l_{\bb}(n,n_0,k_R,k_L,\lal,\bbD,\bbG,Y)}
\newcommand{\Ad}{\mathop{\mathrm{Ad}}\nolimits}
\newcommand{\Proj}{\mathop{\mathrm{Proj}}\nolimits}
\newcommand{\RRe}{\mathop{\mathrm{Re}}\nolimits}
\newcommand{\RIm}{\mathop{\mathrm{Im}}\nolimits}
\newcommand{\Wo}{\mathop{\mathrm{Wo}}\nolimits}
\newcommand{\Prim}{\mathop{\mathrm{Prim}_1}\nolimits}
\newcommand{\Primz}{\mathop{\mathrm{Prim}}\nolimits}
\newcommand{\ClassA}{\mathop{\mathrm{ClassA}}\nolimits}
\newcommand{\Class}{\mathop{\mathrm{Class}}\nolimits}
\newcommand{\diam}{\mathop{\mathrm{diam}}\nolimits}
\def\qed{{\unskip\nobreak\hfil\penalty50
\hskip2em\hbox{}\nobreak\hfil$\square$
\parfillskip=0pt \finalhyphendemerits=0\par}\medskip}
\def\proof{\trivlist \item[\hskip \labelsep{\bf Proof.\ }]}
\def\endproof{\null\hfill\qed\endtrivlist\noindent}
\def\proofof[#1]{\trivlist \item[\hskip \labelsep{\bf Proof of #1.\ }]}
\def\endproofof{\null\hfill\qed\endtrivlist\noindent}
\numberwithin{equation}{section}

%%%%%%%%%%%%%%%%%%%%%%%%%%%%%%%%%%%%%%%
\newcommand{\kakunin}[1]{{\color{red}}}
%%%%%%%%%%%%%%%%%%%%%%%%%%%%%%%%%%%%%%%

\newcommand{\oo}{{\boldsymbol\omega}}
\newcommand{\ctv}{\caC_{(\theta,\varphi)}}
\newcommand{\btv}{\caB_{(\theta,\varphi)}}
\newcommand{\amf}{\caB}
\newcommand{\at}{\caA_{\bbZ^2}}
\newcommand{\oz}{\caO_0}
\newcommand{\caA}{{\mathcal A}}
\newcommand{\caB}{{\mathcal B}}
\newcommand{\caC}{{\mathcal C}}
\newcommand{\caD}{{\mathcal D}}
\newcommand{\caE}{{\mathcal E}}
\newcommand{\caF}{{\mathcal F}}
\newcommand{\caG}{{\mathcal G}}
\newcommand{\caH}{{\mathcal H}}
\newcommand{\caI}{{\mathcal I}}
\newcommand{\caJ}{{\mathcal J}}
\newcommand{\caK}{{\mathcal K}}
\newcommand{\caL}{{\mathcal L}}
\newcommand{\caM}{{\mathcal M}}
\newcommand{\caN}{{\mathcal N}}
\newcommand{\caO}{{\mathcal O}}
\newcommand{\caP}{{\mathcal P}}
\newcommand{\caQ}{{\mathcal Q}}
\newcommand{\caR}{{\mathcal R}}
\newcommand{\caS}{{\mathcal S}}
\newcommand{\caT}{{\mathcal T}}
\newcommand{\caU}{{\mathcal U}}
\newcommand{\caV}{{\mathcal V}}
\newcommand{\caW}{{\mathcal W}}
\newcommand{\caX}{{\mathcal X}}
\newcommand{\caY}{{\mathcal Y}}
\newcommand{\caZ}{{\mathcal Z}}
\newcommand{\bba}{{\mathbb a}}
\newcommand{\bbA}{{\mathbb A}}
\newcommand{\bbB}{{\mathbb B}}
\newcommand{\bbC}{{\mathbb C}}
\newcommand{\bbD}{{\mathbb D}}
\newcommand{\bbE}{{\mathbb E}}
\newcommand{\bbF}{{\mathbb F}}
\newcommand{\bbG}{{\mathbb G}}
\newcommand{\bbH}{{\mathbb H}}
\newcommand{\bbI}{{\mathbb I}}
\newcommand{\bbJ}{{\mathbb J}}
\newcommand{\bbK}{{\mathbb K}}
\newcommand{\bbL}{{\mathbb L}}
\newcommand{\bbM}{{\mathbb M}}
\newcommand{\bbN}{{\mathbb N}}
\newcommand{\bbO}{{\mathbb O}}
\newcommand{\bbP}{{\mathbb P}}
\newcommand{\bbQ}{{\mathbb Q}}
\newcommand{\bbR}{{\mathbb R}}
\newcommand{\bbS}{{\mathbb S}}
\newcommand{\bbT}{{\mathbb T}}
\newcommand{\bbU}{{\mathbb U}}
\newcommand{\bbV}{{\mathbb V}}
\newcommand{\bbW}{{\mathbb W}}
\newcommand{\bbX}{{\mathbb X}}
\newcommand{\bbY}{{\mathbb Y}}
\newcommand{\bbZ}{{\mathbb Z}}
\newcommand{\str}{^*}
\newcommand{\lv}{\left \vert}
\newcommand{\rv}{\right \vert}
\newcommand{\lV}{\left \Vert}
\newcommand{\rV}{\right \Vert}
\newcommand{\la}{\left \langle}
\newcommand{\ra}{\right \rangle}
\newcommand{\ltm}{\left \{}
\newcommand{\rtm}{\right \}}
\newcommand{\lcm}{\left [}
\newcommand{\rcm}{\right ]}
\newcommand{\ket}[1]{\lv #1 \ra}
\newcommand{\bra}[1]{\la #1 \rv}
\newcommand{\lmk}{\left (}
\newcommand{\rmk}{\right )}
\newcommand{\al}{{\mathcal A}}
\newcommand{\md}{M_d({\mathbb C})}
\newcommand{\ainn}{\mathop{\mathrm{AInn}}\nolimits}
\newcommand{\id}{\mathop{\mathrm{id}}\nolimits}
\newcommand{\Tr}{\mathop{\mathrm{Tr}}\nolimits}
\newcommand{\Ran}{\mathop{\mathrm{Ran}}\nolimits}
\newcommand{\Ker}{\mathop{\mathrm{Ker}}\nolimits}
\newcommand{\Aut}{\mathop{\mathrm{Aut}}\nolimits}
\newcommand{\spn}{\mathop{\mathrm{span}}\nolimits}
\newcommand{\Mat}{\mathop{\mathrm{M}}\nolimits}
\newcommand{\UT}{\mathop{\mathrm{UT}}\nolimits}
\newcommand{\DT}{\mathop{\mathrm{DT}}\nolimits}
\newcommand{\GL}{\mathop{\mathrm{GL}}\nolimits}
\newcommand{\spa}{\mathop{\mathrm{span}}\nolimits}
\newcommand{\supp}{\mathop{\mathrm{supp}}\nolimits}
\newcommand{\rank}{\mathop{\mathrm{rank}}\nolimits}
\newcommand{\idd}{\mathop{\mathrm{id}}\nolimits}
\newcommand{\ran}{\mathop{\mathrm{Ran}}\nolimits}
\newcommand{\dr}{ \mathop{\mathrm{d}_{{\mathbb R}^k}}\nolimits} 
\newcommand{\dc}{ \mathop{\mathrm{d}_{\cc}}\nolimits} \newcommand{\drr}{ \mathop{\mathrm{d}_{\rr}}\nolimits} 
\newcommand{\zin}{\mathbb{Z}}
\newcommand{\rr}{\mathbb{R}}
\newcommand{\cc}{\mathbb{C}}
\newcommand{\ww}{\mathbb{W}}
\newcommand{\nan}{\mathbb{N}}\newcommand{\bb}{\mathbb{B}}
\newcommand{\aaa}{\mathbb{A}}\newcommand{\ee}{\mathbb{E}}
\newcommand{\pp}{\mathbb{P}}
\newcommand{\wks}{\mathop{\mathrm{wk^*-}}\nolimits}
\newcommand{\mk}{{\Mat_k}}
\newcommand{\mnz}{\Mat_{n_0}}
\newcommand{\mn}{\Mat_{n}}
\newcommand{\dist}{\dc}
\newcommand{\braket}[2]{\left\langle#1,#2\right\rangle}
\newcommand{\ketbra}[2]{\left\vert #1\right \rangle \left\langle #2\right\vert}
\newcommand{\abs}[1]{\left\vert#1\right\vert}
\newcommand{\trl}[2]
{T_{#1}^{(\theta,\varphi), \Lambda_{#2},\bar V_{#1,\Lambda_{#2}}}}
\newcommand{\trlz}[1]
{T_{#1}^{(\theta,\varphi), \Lambda_{0},\unit}}
\newcommand{\trlt}[2]
{T_{#1}^{(\theta,\varphi), \Lambda_{#2}+t_{#2}\bm e_{\Lambda_{#2}},\bar V_{#1,\Lambda_{#2}+t_{#2}\bm e_{\Lambda_{#2}}}}}
\newcommand{\trltj}[4]
{T_{#1}^{(\theta,\varphi), \Lambda_{#2}^{(#3)}+t_{#4}\bm e_{\Lambda_{#2}^{(#3)}},\bar V_{#1,\Lambda_{#2}^{(#3)}+t_{#4}\bm e_{\Lambda_{#2}^{(#3)}}}}}
\newcommand{\trltjp}[4]
{T_{#1}^{(\theta,\varphi), {\Lambda'}_{#2}^{(#3)}+t_{#4}'\bm e_{{\Lambda'}_{#2}^{(#3)}},\bar V_{#1,{\Lambda'}_{#2}^{(#3)}+t_{#4}'\bm e_{{\Lambda'}_{#2}^{(#3)}}}}}
\newcommand{\trlta}[2]
{T_{#1}^{(\theta,\varphi), \Lambda_{#2}^{t_{#2}},\bar V_{#1,\Lambda_{#2}^{t_{#2}}}}}
\newcommand{\trltb}[2]
{T_{#1}^{(\theta,\varphi), \Lambda_{#2}+t\bm e_{\Lambda_{#2}},\bar V_{#1,\Lambda_{#2}+t\bm e_{\Lambda_{#2}}}}}

\newcommand{\trlpt}[2]
{T_{#1}^{(\theta,\varphi), \Lambda_{#2}'+t_{#2}'\bm e_{\Lambda_{#2}'},\bar V_{#1,\Lambda_{#2}'+t_{#2}'\bm e_{\Lambda_{#2}'}}}}

\newcommand{\trll}[3]
{T_{#1, #3}^{(\theta,\varphi), \Lambda_{#2},\bar V_{#1,\Lambda_{#2}}}}
\newcommand{\trlp}[2]
{T_{#1}^{(\theta,\varphi), \Lambda_{#2}',\bar V_{#1,\Lambda_{#2}'}}}
\newcommand{\trlpp}[2]
{T_{#1}^{(\theta,\varphi), \Lambda_{#2}'',\bar V_{#1,\Lambda_{#2}''}}}
\newcommand{\trlj}[3]
{T_{\rho_{#1}}^{(\theta,\varphi), \Lambda_{#2}^{(#3)}, V_{\rho_{#1},\Lambda_{#2}^{(#3)}}}}
\newcommand{\trljp}[3]
{T_{{\rho'}_{#1}}^{(\theta,\varphi), {\Lambda'}_{#2}^{(#3)},V_{\rho'_{#1},{\Lambda'}_{#2}^{(#3)}}}}
\newcommand{\wod}[3]
{W_{#1\Lambda_{#2}\Lambda_{#3}}}
\newcommand{\wodt}[3]
{{W^{\bm t}}_{#1\Lambda_{#2}\Lambda_{#3}}}
\newcommand{\comp}[2]
{{(\theta_{#1},\varphi_{#1}), \Lambda_{#2},\{\bar V_{\eta,\Lambda_{#2}}\}_\eta}}
\newcommand{\ltj}[2]{\Lambda_{#1}+{#2} \bm e_{\Lambda_{#1}} }
\newcommand{\ltjp}[2]{{\Lambda'}_{#1}+{#2} \bm e_{\Lambda_{#1}} }
\newtheorem{nota}{Notation}[section]
\def\qed{{\unskip\nobreak\hfil\penalty50
\hskip2em\hbox{}\nobreak\hfil$\square$
\parfillskip=0pt \finalhyphendemerits=0\par}\medskip}
\def\proof{\trivlist \item[\hskip \labelsep{\bf Proof.\ }]}
\def\endproof{\null\hfill\qed\endtrivlist\noindent}
\def\proofof[#1]{\trivlist \item[\hskip \labelsep{\bf Proof of #1.\ }]}
\def\endproofof{\null\hfill\qed\endtrivlist\noindent}
%%%%%%%%%%%%%%%%%%%%%%%%%%%%%%
\newcommand{\wrl}[2]{Y_{#1}^{\Lambda_0^{(#2)}}}
\newcommand{\wrlt}[2]{\tilde Y_{#1}^{\Lambda_0^{(#2)}}}
\newcommand{\ZZ}{\bbZ_2\times\bbZ_2}
\newcommand{\SSS}{\mathcal{S}}
\newcommand{\cs}{S}
\newcommand{\ct}{t}
\newcommand{\hS}{S}
\newcommand{\vv}{{\boldsymbol v}}
\newcommand{\ala}{a}
\newcommand{\bet}{b}
\newcommand{\gam}{c}
\newcommand{\alphas}{\alpha}
\newcommand{\alphai}{\alpha^{(\sigma_{1})}}
\newcommand{\alphan}{\alpha^{(\sigma_{2})}}
\newcommand{\betas}{\beta}
\newcommand{\betai}{\beta^{(\sigma_{1})}}
\newcommand{\betan}{\beta^{(\sigma_{2})}}
\newcommand{\alphass}{\alpha^{{(\sigma)}}}
\newcommand{\uu}{V}
\newcommand{\vp}{\varsigma}
\newcommand{\vpr}{R}
\newcommand{\tg}{\tau_{\Gamma}}
\newcommand{\sgg}{\Sigma_{\Gamma}}
\newcommand{\nh}{t28}
\newcommand{\rk}{6}
\newcommand{\nii}{2}
\newcommand{\nhh}{28}
\newcommand{\sjt}{30}
\newcommand{\sjtg}{30}
\newcommand{\bcg}{\caB(\caH_{\alpha})\otimes  C^{*}(\Sigma)}
\newcommand{\pza}[1]{\pi_0\lmk\caA_{\Lambda_{#1}}\rmk''}
\newcommand{\pzac}[1]{\pi_0\lmk\caA_{\Lambda_{#1}}\rmk'}
\newcommand{\pzacc}[1]{\pi_0\lmk\caA_{\Lambda_{#1}^c}\rmk'}
\newcommand{\trlzi}[2]{T_{#1}^{(\theta,\varphi) \Lambda_0^{(#2)}\unit}}

%%%%%%%%%%%%%%%%%%%%%%%%%%%%%%%%%%%%%%%
\newcommand{\obk}{\omega_{\mathop{\mathrm{bk}}}}
\newcommand{\obd}{\omega_{\mathop{\mathrm{bd}}}}
\newcommand{\obdm}{\omega_{\mathop{\mathrm{bd}(-)}}}
\newcommand{\abk}{\caA_{\mathbb Z^2}}
\newcommand{\hu}{\mathop {\mathrm H_{U}}}
\newcommand{\hd}{\mathop {\mathrm H_{D}}}

\newcommand{\abd}{\caA_{\hu}}
\newcommand{\aloch}{\caA_{\hu,\mathop{\mathrm {loc}}}}
\newcommand{\alocg}[1]{\caA_{#1,\mathop{\mathrm {loc}}}}
\newcommand{\hbk}{\caH_{\mathop{\mathrm{bk}}}}
\newcommand{\hbd}{\caH_{\mathop{\mathrm{bd}}}}
\newcommand{\hbdm}{\caH_{\mathop{\mathrm{bd}(-)}}}
\newcommand{\pbk}{\pi_{\mathop{\mathrm{bk}}}}
\newcommand{\pbd}{\pi_{\mathop{\mathrm{bd}}}}
\newcommand{\pbdm}{\pi_{\mathop{\mathrm{bd}(-)}}}
\newcommand{\mopbk}{{\mathop{\mathrm{bk}}}}
\newcommand{\mopbd}{{\mathop{\mathrm{bd}}}}
\newcommand{\Obk}{O_{\mathop{\mathrm{bk}}}}
\newcommand{\OUbk}{O_{\mathop{\mathrm{bk}}}^U}
\newcommand{\Orbd}{O^r_{\mathop{\mathrm{bd}}}}
\newcommand{\OrUbd}{O^{rU}_{\mathop{\mathrm{bd}}}}
\newcommand{\Obkl}{O_{\mathop{\mathrm{bk},\Lambda_0}}}
\newcommand{\OUbkl}{O_{\mathop{\mathrm{bk},\Lambda_0}}^U}
\newcommand{\Orbdl}{O^r_{\mathop{\mathrm{bd},\Lambda_{r0}}}}
\newcommand{\OrUbdl}{O^{rU}_{\mathop{\mathrm{bd},\Lambda_{r0}}}}
\newcommand{\Obu}{O_{\mathop{\mathrm{bd}}}^{\mathop{\mathrm{bu}}}}
\newcommand{\Obul}{O_{\mathop{\mathrm{bd}},\lz}^{\mathop{\mathrm{bu}}}}
\newcommand{\Odl}{O_{\lzr}^{\caD}}
\newcommand{\fd}{F^{\caD}}
\newcommand{\hilb}{\mathop{\mathrm {Hilb}}_f}
\newcommand{\Obun}[1]{O_{\mathop{\mathrm{bd}#1}}^{\mathop{\mathrm{bu}}}}
\newcommand{\OrUbdn}[1]{O^{rU}_{\mathop{\mathrm{bd}#1 }}}
%%%%%%%%%%%%%%%%%%%%%%%%%%%%%%%%%%%%%%%%%%%%%%%%%
\newcommand{\OUbkm}{O_{\mathop{\mathrm{bk}}}^{U(-)}}
\newcommand{\Orbdm}{O^{r(-)_{\mathop{\mathrm{bd} }}}}
\newcommand{\OrUbdm}{O^{rU(-)}_{\mathop{\mathrm{bd}}}}
\newcommand{\OUbklm}{O_{\mathop{\mathrm{bk},\Lambda_{0(-)}}}^{U(-)}}
\newcommand{\Orbdlm}{O^{r(-)}_{\mathop{\mathrm{bd},\Lambda_{r0(-)}}}}
\newcommand{\OrUbdlm}{O^{rU(-)}_{\mathop{\mathrm{bd},\Lambda_{r0(-)}}}}
\newcommand{\Obum}{O_{\mathop{\mathrm{bd}}}^{\mathop{\mathrm{bu}(-)}}}
\newcommand{\Obulm}{O_{\mathop{\mathrm{bd}},{\lm {0(-)}}}^{\mathop{\mathrm{bu}(-)}}}
\newcommand{\Otot}{O_{\mathop{\mathrm{total}}}}
\newcommand{\Ototl}{O_{\mathop{\mathrm{total}}}^{\lz,\lzm}}

\newcommand{\Obj}{{\mathop{\mathrm{Obj}}}}
\newcommand{\Mor}{{\mathop{\mathrm{Mor}}}}
\newcommand{\ti}[4]{\tilde\iota\lmk(#1,#2), (#3,#4)\rmk}

%%%%%%%%%%%%%%%%%%%%%%%%%%%%%%%%%%%%%%%%%%%%%%%%%
\newcommand{\Cabkl}{C_{\mathop{\mathrm{bk},\lz}}}
\newcommand{\Cabk}{C_{\mathop{\mathrm{bk}}}}
\newcommand{\CaUbk}{C_{\mathop{\mathrm{bk}}}^U}
\newcommand{\Carbd}{C^r_{\mathop{\mathrm{bd}}}}
\newcommand{\CarUbd}{C^{rU}_{\mathop{\mathrm{bd}}}}
\newcommand{\CaUbkl}{C_{\mathop{\mathrm{bk},\Lambda_0}}^U}
\newcommand{\Carbdl}{C^r_{\mathop{\mathrm{bd},\Lambda_{r0}}}}
\newcommand{\CarUbdl}{C^{rU}_{\mathop{\mathrm{bd},\Lambda_{r0}}}}
\newcommand{\Cabu}{C_{\mathop{\mathrm{bd}}}^{\mathop{\mathrm{bu}}}}
\newcommand{\Cabul}{C_{\mathop{\mathrm{bd}},\lz}^{\mathop{\mathrm{bu}}}}
\newcommand{\Cadl}{C_{\lzr}^{\caD}}
\newcommand{\Hom}{\mathop{\mathrm{Hom}}}
\newcommand{\Catotl}{C_{\mathop{\mathrm{total}}}^{\lz,\lzm}}
\newcommand{\Cafin}{C_{\mathop{\mathrm{bd}},\lz}^{\mathop{\mathrm{fin}}}}

%%%%%%%%%%%%%%%%%%%%%%%%%%%%%%%%%%%%%%%%%%%%%%%%%%
\newcommand{\mm}{\Mor_{\tilde\caM}}
\newcommand{\om}{\Obj\lmk {\tilde\caM}\rmk}
\newcommand{\omt}{\otimes_{\tilde\caM}}
%%%%%%%%%%%%%%%%%%%%%%%%%%%%%%%%%%%%%%%%%%%%%%%%%%

\newcommand{\CaUbkm}{C_{\mathop{\mathrm{bk}}}^{U(-)}}
\newcommand{\Carbdm}{C^{r(-)}_{\mathop{\mathrm{bd}}}}
\newcommand{\CarUbdm}{C^{rU(m)}_{\mathop{\mathrm{bd}}}}
\newcommand{\CaUbklm}{C_{\mathop{\mathrm{bk},\Lambda_{0(-)}}}^{U(-)}}
\newcommand{\Carbdlm}{C^{r(-)}_{\mathop{\mathrm{bd},\Lambda_{r0(-)}}}}
\newcommand{\CarUbdlm}{C^{rU(-)}_{\mathop{\mathrm{bd},\Lambda_{r0(-)}}}}
\newcommand{\Cabum}{C_{\mathop{\mathrm{bd}}}^{\mathop{\mathrm{bu}(-)}}}
\newcommand{\Cabuln}{C_{\mathop{\mathrm{bd}},\lm{0(-)}}^{\mathop{\mathrm{bu}(-)}}}
\newcommand{\Irr}{\mathop{\mathrm{Irr}}}

%%%%%%%%%%%%%%%%%%%%%%%%%%%%%%%%%%%%%%%%%%%%%%%%%%

\newcommand{\Tbk}[3]{T_{#1}^{(\frac{3\pi}2,\frac\pi 2), {#2},{#3}}}
\newcommand{\Vrl}[2]{V_{#1,#2}}
\newcommand{\Tbkv}[2]{T_{#1}^{(\frac{3\pi}2,\frac\pi 2), #2, V_{#1,#2}}}
\newcommand{\Tbd}[3]{T_{#1}^{\mathrm{(l)} #2 #3}}
\newcommand{\Tbdv}[2]{T_{#1}^{\mathrm{(l)} #2\Vrl{{#1}}{#2}}}

\newcommand{\zc}{\caZ\lmk\Cabul\rmk}
\newcommand{\hi}[2]{\hat\iota\lmk #1: #2\rmk}

\newcommand{\lz}{\Lambda_0}
\newcommand{\lzr}{\Lambda_{r0}}
\newcommand{\lm}[1]{{\Lambda_{#1}}}
\newcommand{\llz}{(\lz,\lzr)}
\newcommand{\lmr}[1]{{\Lambda_{r#1}}}
\newcommand{\hlm}[1]{{\hat\Lambda_{#1}}}
\newcommand{\tlm}[1]{{\tilde\Lambda_{#1}}}
\newcommand{\ld}{\Lambda}
\newcommand{\pc}{\mathcal{PC}}
\newcommand{\Cbk}{\caC_{\mopbk}}
\newcommand{\CUbk}{\caC_{\mopbk}^U}
\newcommand{\Crbd}{\caC_{\mopbd}^r}
\newcommand{\Clbd}{\caC_{\mopbd}^l}
\newcommand{\Vbk}[2]{\caV_{#1#2}^{\mopbk}}
\newcommand{\Vbd}[2]{\caV_{#1#2}^{\mopbd}}
\newcommand{\VUbd}[2]{\caV_{#1#2}^{\mopbd U} }
\newcommand{\Vbu}[2]{\caV_{#1#2}^{\mathop{\mathrm{bu}}}}
\newcommand{\bl}{\caB_l}
\newcommand{\fbk}{\caF_{\mopbk}^U}
\newcommand{\fbd}{\caF_{\mopbd}^U}
\newcommand{\gu}{\caG^U}
\newcommand{\hb}[2]{\iota^{(\lz,\lzr)}\lmk#1: #2\rmk}
\newcommand{\hfc}{\hat F^{\llz}}
\newcommand{\ffc}{F^{\llz}}

%%%%%%%%%%%%%%%%%%%%%%%%%%%%%%%%%%%%%%%%%
\newcommand{\lzm}{\Lambda_{0(-)}}
\newcommand{\lzrm}{\Lambda_{r0(-)}}
\newcommand{\pcm}{\mathcal{PC(-)}}
\newcommand{\CUbkm}{\caC_{\mopbk}^{U(-)}}
\newcommand{\Crbdm}{\caC_{\mopbd}^{r(-)}}
\newcommand{\Clbdm}{\caC_{\mopbd}^{l(-)}}
\newcommand{\Vbkm}[2]{\caV_{#1#2}^{\mopbk (-)}}
\newcommand{\Vbdm}[2]{\caV_{#1#2}^{\mopbd(-)}}
\newcommand{\VUbdm}[2]{\caV_{#1#2}^{\mopbd U(-)} }
\newcommand{\Vbum}[2]{\caV_{#1#2}^{\mathop{\mathrm{bu}(-)}}}
\newcommand{\fbkm}{\caF_{\mopbk}^{U(-)}}
\newcommand{\fbdm}{\caF_{\mopbd}^{U(-)}}
\newcommand{\vrdh}[2]{\caV_{#1\ld_{#2}}}
\newcommand{\vrd}{\caV_{\rho\ld}}
\newcommand{\Vrd}{V_{\rho\ld}}
\newcommand{\tvrd}{\tilde\caV_{\rho\ld}}

%%%%%%%%%%%%%%%%%%%%%%%%%%%%%%%%%%%%%%%%%
\newcommand{\lr}[2]{\Lambda_{(#1,0),#2,#2}}
\newcommand{\lef}[2]{\overline{\Lambda_{(#1,0),\pi-#2,#2}}}
\newcommand{\lrhu}{(\Lambda_r)^c\cap \hu}
\newcommand{\lhuc}{\Lambda^c\cap \hu}
\newcommand{\lhu}{\Lambda\cap \hu}
\newcommand{\lrhuz}{(\Lambda_{r0})^c\cap \hu}
\newcommand{\lhucz}{\Lambda_{0}^c\cap \hu}
\newcommand{\lc}[1]{(\lm #1)^c\cap\hu}

%%%%%%%%%%%%%%%%%%%%%%%%%%%%%%%%%%%%%%%
\newcommand{\zam}{\caZ_a\lmk\tilde\caM\rmk }
\newcommand{\ozam}{\Obj\lmk \zam\rmk}
\newcommand{\mzam}{\Mor_{\zam}}
\newcommand{\rpc}[2]{\lmk(#1,#2), C_{(#1,#2)}\rmk}
\newcommand{\ozt}{\otimes_{\zam}}
\newcommand{\ool}{\caO_{\omega,\ld_0}}
\newcommand{\ools}[1]{\caO_{#1,\ld_0}}
\newcommand{\col}{C_{\omega,\ld_0}}
\newcommand{\ctvz}{\caC_{(\theta_0,\varphi_0)}}
\newcommand{\cols}[1]{C_{#1,\ld_0}}
\newcommand{\coll}[1]{C_{\omega,#1}}
\newcommand{\tool}{\tilde \caO_{\omega,\ld_0}}
\newcommand{\tcols}[1]{\tilde C_{#1,\ld_0}}
\newcommand{\tcol}{\tilde C_{\omega,\ld_0}}
\newcommand{\folz}{F_{\omega,\varphi,\ld_0}}
\newcommand{\fol}[2]{F_{#1,#2,\ld_0}}

%%%%%%%%%%%%%%%%%%%%%%%%%%%%%%%%%%%%%%%%%%%
\newcommand{\change}[1]{{\color{red}{#1}}}
\title{Mixed state topological order: operator algebraic approach}
\author{Yoshiko Ogata \\Research Institute for Mathematical Sciences\\
 Kyoto University, Kyoto 606-8502 JAPAN}
\maketitle
\begin{abstract}
We study the classification problem of mixed states in two-dimensional quantum spin systems in the operator algebraic framework of quantum statistical mechanics.
We associate a braided $C^*$-tensor category to each state satisfying a mixed-state version of the approximate Haag duality.
We study how this category behaves under decoherence: suppose the state is acted by a finite depth quantum channel. We prove that the braided $C^*$-tensor category of the final state  is a braided $C^*$-tensor subcategory of  the initial state.

\end{abstract}
\section{Introduction}
The classification of topological phases of matter has attracted a lot of attention in the last two decades. As a result, many things are now known about  2-dimensional gapped ground states, i.e., on isolated zero-temperature pure states \cite{wen2004quantum}\cite{kitaev2006anyons}\cite{wen2016theory}\cite{wen2017colloquium}. Some of them are even proven mathematically rigorously. In contrast, although there is good progress \cite{dennis2002topological}\cite{konig2014generating}\cite{coser2019classification}\cite{lu2020detecting}\cite{sohal2024noisy}\cite{ellison2024towards}, our knowledge of the topological phases of open systems is still quite limited, compared to that of  isolated systems. Decoherence due to interactions with the environment is inevitable, and it is important to study the classification problem of mixed states.

In this paper, we consider this problem from the operator algebraic point of view. The operator algebraic approach to topological order using algebraic quantum field theory machinery was initiated by Naaijkens in his seminal study of Toric code and abelian quantum double models \cite{N2}\cite{N1}\cite{naaijkens2013kosaki}\cite{FN}. The advantage of this natural approach is that in this framework, we can show that anyons are invariant of the classification problem of 2-d gapped ground state phases under the assumption called approximate Haag duality\cite{MTC}. 

We introduce a mixed-state version of \cite{MTC} and derive a braided $C^*$-tensor category $\cols{\omega\otimes\psi}$ out of each state $\omega$ satisfying a mixed-state version of the approximate Haag duality. We study how this category behaves under decoherence: suppose that a state $\omega_1$ interacts with its environment (which is set to be a trivial state at time zero) for a finite time or acted by a finite depth quantum channel. Suppose that the initial state $\omega_1$ and the final state $\omega_2$ both satisfy the mixed version approximate Haag duality. We prove that the category of the final state $\cols{\omega_2\otimes\psi}$ is a subcategory of the category of the initial state $\cols{\omega_1\otimes\psi}$. Note that 
$\cols{\omega_2\otimes\psi}$ may not be a {\it full} subcategory of $\cols{\omega_1\otimes\psi}$. It means there can be two quasi-particles that are isomorphic to each other in $\cols{\omega_1\otimes\psi}$ but look distinct in $\cols{\omega_2\otimes\psi}$.

\subsection{$2$-dimensional quantum spin systems}
Now, we introduce our concrete setting. For basic notation, see Appendix \ref{notation}.
By a $2$-dimensional quantum spin system, we mean a $C^*$-algebra constructed as follows.
We denote the algebra of $d\times d$ matrices by $\Mat_{d}$.
For each $z\in\bbZ^2$,  let $\caA_{\{z\}}$ be an isomorphic copy of $\Mat_{d}$, and for any finite subset $\Lambda\subset\bbZ^2$, we set $\caA_{\Lambda} = \bigotimes_{z\in\Lambda}\caA_{\{z\}}$.
%, which is the local algebra of observables in $\Lambda$. 
For finite $\Lambda$, the algebra $\caA_{\Lambda} $ can be regarded as the set of all bounded operators acting on
the Hilbert space $\bigotimes_{z\in\Lambda}{\bbC}^{d}$.
We use this identification freely.
If $\Lambda_1\subset\Lambda_2$, the algebra $\caA_{\Lambda_1}$ is naturally embedded in $\caA_{\Lambda_2}$ by tensoring its elements with the identity. 
For an infinite subset $\Gamma\subset \bbZ^{2}$,
$\caA_{\Gamma}$
is given as the inductive limit of the algebras $\caA_{\Lambda}$ with $\Lambda$, finite subsets of $\Gamma$.
We call $\caA_{\Gamma}$ the quantum spin system on $\Gamma$.
The two-dimensional quantum spin system is the algebra $\caA_{\bbZ^2}$,
which is denoted by $\caA$ as well.

For a subset $\Gamma_1$ of $\Gamma\subset\bbZ^{2}$,
the algebra $\caA_{\Gamma_1}$ can be regarded as a subalgebra of $\caA_{\Gamma}$. 
For $\Gamma\subset \bbR^2$, with a bit of abuse of notation, we write $\caA_{\Gamma}$
to denote $\caA_{\Gamma\cap \bbZ^2}$.
Also, $\Gamma^c$ denotes the complement of $\Gamma$ in $\bbR^2$.
A representation $\rho$ of $\caA$ on a Hilbert space $\caH$ is a nonzero $*$-homomorphism from
$\caA$ to $\caB(\caH)$, which is not necessarily unital, in this paper. 

In our framework, regions called cones play an important role.
For
each $\bm a\in \bbR^2$, $\theta\in\bbR$ and $\varphi\in (0,\pi)$,
we set 
\begin{align*}
\Lambda_{\bm a, \theta,\varphi}
:
=&\lmk \bm a+\left\{
t\bm e_{\beta}\mid t>0,\quad \beta\in (\theta-\varphi,\theta+\varphi)\right\}\rmk
\end{align*}
where  $\bm e_{\beta}=(\cos\beta,\sin\beta)$.
We call a set of this form a cone.
We set $\arg\Lambda_{\bm a, \theta,\varphi}:=[\theta-\varphi,\theta+\varphi]$, where the
 right hand side should be understood $\mod 2\pi$.
We also set $|\arg\Lambda|=2\varphi$.
and $\bm e_\Lambda:=\bm e_{\theta}$ for $\Lambda=\Lambda_{\bm a, \theta,\varphi}$.
For $\varepsilon>0$, $t\in \bbR$ and $\Lambda=\Lambda_{\bm a, \theta,\varphi}$, $\Lambda_\varepsilon$ 
denotes $\Lambda_\varepsilon=\Lambda_{\bm a, \theta,\varphi+\varepsilon}$,
$\Lambda(t):=\ld +t {\bm e}_\ld$.
For
$\theta_0\in\bbR$ and $\varphi_0\in (0,\pi)$ we consider the following set of cones:
\begin{align}
\begin{split}
&\caC_{(\theta_0,\varphi_0)}:=\left\{
\Lambda : \text{cone}\mid \arg\Lambda\cap{[\theta_0-\varphi_0,\theta_0+\varphi_0]}=\emptyset
\right\}.
\end{split}
\end{align}
Again, the equality in the parentheses should be understood $\mod 2\pi$.
Note that $\caC_{(\theta_0,\varphi_0)}$ is an upward-filtering set with respect to the inclusion relation.
For each $\Lambda_i\in \caC_{(\theta_0,\varphi_0)}$ 
$i=1,2$,  we write $\Lambda_2\leftarrow_{(\theta_0,\varphi_0)}\Lambda_1$
if $\Lambda_i=\Lambda_{\bm {a}_i,\theta_i,\varphi_i}$ with
\[
\theta_0+\varphi_0<
\theta_1-\varphi_1<\theta_1+\varphi_1<\theta_2-\varphi_2<\theta_2+\varphi_2
<\theta_0-\varphi_0+2\pi.
\]
\subsection{Mixed state Approximate Haag duality}
We introduce an approximate version of the relative Haag duality introduced in \cite{camassa2007relative}.
This corresponds to a mixed state state version of the approximate Haag duality in \cite{MTC}.
\begin{defn}\label{AHdef}
Let $\caA$ be a $2$-dimensional quantum spin system.
Let $\omega$ be a state on $\caA$
with a GNS representation $(\caH,\pi)$.
We say $\omega$ satisfies the approximate Haag duality 
if, for any $\zeta\in (0,\pi)$, $0<\varepsilon<\frac14(\pi-\zeta)$,
there exists an $R_{\zeta,\varepsilon}\ge 0$
and a decreasing function $f_{\zeta, \varepsilon}: \bbR_{\ge 0}\to \bbR_{\ge 0}$
with $\lim_{t\to\infty} f_{\zeta, \varepsilon}(t)=0$
satisfying the following : 
for any cone $\Lambda$ with $\lv\arg \Lambda\rv=2\zeta$,
there exists a unitary $U_{\Lambda,\varepsilon}\in \caU\lmk\pi\lmk\caA\rmk''\rmk$
such that
\begin{description}
\item[(i)]
\[
\pi(\caA_{\Lambda^c})'\cap \pi\lmk \caA_{\bbZ^2}\rmk''
\subset \Ad \lmk U_{\Lambda,\varepsilon}\rmk\lmk
\pi\lmk\caA_{\Lambda_{\varepsilon}(-R_{\zeta,\varepsilon})}\rmk''
\rmk,\]
 and
\item[(ii)]
for any $t\ge 0$, there exists a unitary $U_{\Lambda,\varepsilon, t}\in \caU\lmk
\pi\lmk \caA_{\Lambda_{2\varepsilon}(-t)}\rmk''\rmk$
such that
\begin{align}
\lV U_{\Lambda,\varepsilon, t}-U_{\Lambda,\varepsilon}\rV
\le f_{\zeta,\varepsilon}(t).
\end{align}
\end{description}
\end{defn}
 \subsection{Cone von Neumann algebras}\label{uni}
Let $\omega$ be a state on a $2$-dimensional quantum spin system $\caA$ with a GNS representation $(\caH,\pi)$.
For each cone $\ld$ in $\bbZ^2$, we consider the local von Neumann algebra
$\pi(\caA_{\ld})''$.
With a bit of abuse of notation, we say that $\omega$ has properly infinite cone algebras
if $\pi(\caA_{\ld})''$ is properly infinite for any cone $\ld$.
By the stabilization procedure, namely, by tensoring a pure infinite tensor product state,
we can always move to a state with properly infinite cone algebras:
\begin{lem}\label{yuuyake}
Let $\omega$ be a state on a $2$-dimensional quantum spin system $\caA$.
Let $\caB$ be another two-dimensional quantum spin system and
$\psi$ a pure infinite tensor product state on $\caB$.
Then the state $\omega\otimes \psi$ on $\caA\otimes\caB$ has properly infinite cone algebras.
\end{lem}
\begin{proof}
If $(\caH,\pi)$, $(\caH_\psi,\pi_\psi)$ are GNS representations of $\omega$, $\psi$ respectively,
then $(\caH\otimes \caH_\psi, \pi\otimes \pi_\psi )$ is a GNS representation of
$\omega\otimes \psi$.
Because $\psi$ is a pure infinite tensor product state, the cone von Neumann algebra associated to
a cone $\ld$ is
of the form
$
\pi(\caA_\ld)''\bar\otimes \pi_\psi(\caB_\ld)''
$
with $\pi_\psi(\caB_\ld)''$ a type ${\mathrm I}_\infty$-von Neumann algebra.
The center of this algebra is 
$
\caZ\lmk \pi(\caA_\ld)''\rmk \bar\otimes \bbC\unit
$
(see IV Corollary 5.11 \cite{takesaki}).
Because $\pi_\psi(\caB_\ld)''$ is an infinite factor,
there exists an isometry $v\in\pi_\psi(\caB_\ld)''$ with $vv^*\lneq\unit$.
Therefore, for any nonzero central projection $z\otimes \unit \in \caZ\lmk \pi(\caA_\ld)''\rmk \bar\otimes \bbC\unit$,
we have
\begin{align}
\begin{split}
(z\otimes v)^*(z\otimes v)=z\otimes \unit,\quad 
(z\otimes v)(z\otimes v)^*=z\otimes vv^*\lneq z\otimes\unit.
\end{split}
\end{align}
Because $z\otimes v\in \pi(\caA_\ld)''\bar\otimes \pi_{\psi}(\caB_\ld)''$, this means $z\otimes \unit$ is infinite
in $\pi(\caA_\ld)''\bar\otimes \pi_{\psi}(\caB_\ld)''$. Hence $\pi(\caA_\ld)''\bar\otimes \pi_\psi(\caB_\ld)''$
is properly infinite.
\end{proof}
Therefore, by tensoring a pure infinite tensor product state, we can always make a
state with properly infinite cone algebras.

%We also set $\caA_{\rm loc}:=\bigcup_{\Lambda\in{\mathfrak S}_{\bbZ^{2}}}\caA_{\Lambda}
%$.
%We denote by $\beta_x$ the automorphisms on $\caA$ representing the space translation by  $x\in\bbZ$.
%By $Q^{(j)}$, $j\in\bbZ$, we denote the element of $\caA$ with $Q\in\Mat_d$ in the $j$-th component of the tensor product of $\caA$ and the unit in any other component.
\subsection{Mixed state super selection sector}
We consider the following notion of superselection criterion.
\begin{defn}\label{iwashi}
Let $\caA$ be a $2$-dimensional quantum spin system.
Let $\omega$ be a state on $\caA$ with properly infinite cone algebras.
Let $(\caH,\pi)$ be a GNS representation.
For each representation $\rho$ of $\caA$ on $\caH$ and a cone $\ld$,
we denote by $\caV_{\rho\ld}$ the set of all $V_{\rho\ld}\in \pi\lmk\caA\rmk''$
satisfying $\Vrd^*\Vrd\in\pi\lmk \caA_{\ld^c}\rmk'$ and
\begin{align}\label{sss}
\begin{split}
\left. \Ad \Vrd \circ\pi\right\vert_{\caA_{\ld^c}}=\left. \rho\right\vert_{\caA_{\ld^c}}.
\end{split}
\end{align}
We denote by $\caO_\omega$ the set of all representations $\rho$ of $\caA$ on $\caH$,
which have a nonempty $\vrd$ for all cones $\ld$.
For a cone $\ld_0$, the set of all $\rho\in \caO_\omega$ satisfying
\begin{align}\label{take}
\begin{split}
\rho(A)=\pi(A)\rho(\unit),\quad A\in \caA_{\ld_0^c}
\end{split}
\end{align}
is denoted by $\ool$
 \end{defn}
 \begin{rem}\label{panda}
 By the definition, we have
 $
\rho(\caA)\subset \pi(\caA)''
$
for any $\rho\in \caO_\omega$.
 If $\rho\in\ool$, we have 
\begin{align}
\begin{split}
\vrd\subset \pi\lmk\caA_{(\ld\cup\ld_0)^c}\rmk' \cap \pi\lmk\caA\rmk'',
\end{split}
\end{align}
for any cone $\ld$.
In particular, we have $\rho(\unit)\in \pi(\caA_{\ld_0^c})'$.
\end{rem}
\begin{rem}\label{namako}
The superselection criterion (\ref{sss}) is not suitable when cone algebras are not properly infinite (see subsection \ref{uni}).
In this paper, it is not a problem because we will always consider states with properly infinite cone algebras.
One important observation is that 
any state can be stabilized into a state with properly infinite cone algebras
(Lemma \ref{yuuyake}) and the category $\cols{\omega\otimes\psi}$ that we consider is defined for states after stabilization.
\end{rem}
For the latter use, we introduce 
a variation of Definition \ref{iwashi}:
\begin{defn}Let $\caA$ be a two-dimensional quantum spin system.
Let $\omega$ be a state on $\caA$ with a GNS representation $(\caH,\pi)$.
For each representation $\rho$ of $\caA$ on $\caH$ and a cone $\ld$,
we denote by $\tilde \caV_{\rho\ld}$ the set of all $V_{\rho\ld}\in  \caV_{\rho\ld}$
satisfying $\Vrd^*\Vrd=\Vrd\Vrd^*=\rho(\unit)$. 
We denote by $\tilde \caO_\omega$ the set of all representations $\rho$ of $\caA$ on $\caH$
with a nonempty $\tvrd$ for any cone $\ld$, satisfying
$\rho(\unit)\in Z\lmk {\pi(\caA)''}\rmk$.
We also set $\tool:=\tilde \caO_\omega\cap \ool$ for a cone $\ld_0$.
 \end{defn}
 \begin{rem}
 Note by definition we have $\tool\subset \ool$ and $\tvrd\subset\vrd$.
 \end{rem}

\subsection{Main Result}
By the analogous argument as in \cite{MTC} which follows the AQFT recipe \cite{DHRI},
\cite{frohlich1990braid},\cite{FRS},
\cite{longo1990index},\cite{BF}\cite{BDMRS},\cite{longo1990index} we obtain a braided $C^*$-tensor category.
\begin{thm}\label{colder}
Let $\caA$ be a $2$-dimensional quantum spin system.
Let $\omega$ be a state on  $\caA$ with properly infinite cone algebras,
 satisfying the approximate Haag duality.
Let $\ld_0$ be a cone.
Then there exists a braided $C^*$-tensor category $\col$ with objects $\ool$.
If $\ld_0'$ is another cone, then $\col$ and $\coll{\ld_0'}$
are equivalent as braided $C^*$-tensor categories.
\end{thm}
A more detailed statement will be given in section \ref{derisec} Theorem \ref{same}.
Under the stabilization procedure, the categories get saturated:
\begin{thm}\label{stab}
Let $\caA$ be a $2$-dimensional quantum spin system.
Let $\omega$ be a state on $\caA$ with properly infinite cone algebras
 satisfying the approximate Haag duality.
Let $\ld_0$ be a cone.
Let $\caB_1$, $\caB_2$ be two-dimensional quantum spin systems and $\psi_1$, $\psi_2$
pure infinite tensor product states on $\caB_1$, $\caB_2$, respectively.
Then $\cols{\omega\otimes\psi_1}$ and $\cols{\omega\otimes\psi_1\otimes \psi_2}$
are equivalent as braided $C^*$-tensor categories.
\end{thm}
This Theorem also tells us that the choice of $\psi$ doesn't matter for 
$\cols{\omega\otimes\psi}$:
\begin{cor}\label{yofuke}
Let $\caA$ be a $2$-dimensional quantum spin system.
Let $\omega$ be a state on   $\caA$ with properly infinite cone algebras
 satisfying the approximate Haag duality.
Let $\ld_0$ be a cone.
Let $\caB$ be a two-dimensional quantum spin system and $\psi_1$, $\psi_2$
pure infinite tensor product states on $\caB$.
Then $\cols{\omega\otimes\psi_1}$ and $\cols{\omega\otimes \psi_2}$
are equivalent as braided $C^*$-tensor categories.
\end{cor}
Corollary \ref{yofuke} says that we can associate a braided $C^*$-tensor category
$\cols{\omega\otimes\psi}$
to each $\omega$ uniquely, modulo equivalence of the braided $C^*$-tensor categories.

When $\omega$ is pure, the stabilization doesn't change the category:
\begin{prop}\label{kapibara}
Let $\caA$ be a $2$-dimensional quantum spin system.
Let $\omega$ be a pure state on   $\caA$ with properly infinite cone algebras
 satisfying the approximate Haag duality.
Let $\ld_0$ be a cone.
Let $\caB$ be a two-dimensional quantum spin system and $\psi$ a 
pure infinite tensor product states on $\caB$.
Then $\cols{\omega}$ and $\cols{\omega\otimes\psi}$
are equivalent as braided $C^*$-tensor categories.
\end{prop}
The main result of this paper is that when the system interacts with 
environments (which is set to be a trivial state at time zero) for a finite time interval, or acted by finite depth quantum channels,
the category $\cols{\omega_2\otimes\psi}$ of the final state $\omega_2$
 becomes a subcategory of the category $\cols{\omega_1\otimes\psi}$ of the initial state.
\begin{thm}\label{main}
Let $\caA$ be a $2$-dimensional quantum spin system.
Let $\omega_1$, $\omega_2$ be 
states on  $\caA$ with properly infinite cone algebras, satisfying the approximate Haag duality.
Let $\ld_0$ be a cone.
Let $\caB_1$, $\caB_2$ be two-dimensional quantum spin systems and $\psi_1$, $\psi_2$ be
pure infinite tensor product states on $\caB_1$, $\caB_2$ respectively.
Let $\alpha$ be an approximately-factorizable 
automorphism on $\caA\otimes \caB$ in the sense of Definition 1.2 of \cite{MTC}
(See Definition \ref{qfdef} for the precise definition).
Suppose that $\omega_2=\lmk \omega_1\otimes\psi_1\rmk\circ\alpha \vert_{\caA}$.
Then there is a faithful braided tensor functor from
$\cols{\omega_2\otimes \psi_2}$ to $\cols{\omega_1\otimes \psi_2}$.
\end{thm}

Note that the functor is faithful but not necessarily fully faithful.
In particular, it might be possible that two objects which are distinct
in $\cols{\omega_2\otimes \psi_2}$ get isomorphic in $\cols{\omega_1\otimes \psi_2}$.
Physically, it means the following for two anyons $\rho,\sigma$:
If we are allowed to use the total system $\caA\otimes\caB$, $\rho$ and $\sigma$ can be mapped to each other (or identified
 as the same species).
However, if we are allowed to use only the sub-system $\caA$, then it may be impossible
to do so. 

Recall that automorphisms given by (possibly  time-dependent) local interactions are  approximately-factorizable
\cite{Ogata3}.
In particular, finite depth quantum circuits satisfy this property.
Therefore, by the Stinespring dilation theorem, the above theorem covers the case that
$\omega_2=\omega_1\Phi$, with $\Phi$ a finite-depth quantum channel.

This paper is organized as follows. In section \ref{derisec}, we derive braided C*tensor categories out of states with properly infinite cone algebras, satisfying the approximate Haag duality. For pure states, the category is equivalent to that obtained in \cite{MTC}.
In section \ref{subsys}, we analyze the relation between the braided C*-tensor categories when we consider subsystems. In section  \ref{stabilization},  we show Theorem \ref{stab}, the stabilization result.
In section \ref{mainprf}, we give the proof of our main theorem Theorem \ref{main}.

\section{Derivation of braided $C^*$-tensor category}\label{derisec}
In this section, we give the detailed version of Theorem \ref{colder}.
Recall the definition of braided $C^*$-tensor categories in \cite{NT}.
By the analogous argument as in \cite{MTC}, we obtain the following Theorem.
\begin{thm}\label{same}
Let $\caA$ be a $2$-dimensional quantum spin system.
Let $\omega$ be a state on  $\caA$ with properly infinite cone algebras satisfying the approximate Haag duality.
Let $(\caH,\pi)$ be a GNS representation of $\omega$.
Let $\theta_0\in \bbR$, $0<\varphi_0<\pi$ and $\ld_0\in \ctvz$.
Then the following hold.
\begin{description}
\item[(i)]
By setting objects
\begin{align}
\Obj \col := \ool
\end{align}
 and morphisms between objects $\rho,\sigma\in \Obj \col $
 \begin{align}
 \begin{split}
 \Mor_{\col}\lmk\rho,\sigma\rmk
 :=
 \left\{
 R\in \pi(\caA)''\middle |
 \begin{gathered}
 R\rho(A)=\sigma(A) R,\quad A\in \caA,\\
 \sigma(\unit) R\rho(\unit)=R
 \end{gathered}
 \right\}.
 \end{split}
 \end{align}
 with identity morphisms $\id_\rho:=\rho(\unit)$,
 we obtain a $C^*$-category $\col$.
 \item[(ii)]
 For each $\rho\in \ool$, there is a unique endomorphism $T_\rho$ on
 \begin{align}
 \caF:=\overline{\cup_{\ld\in \caC_{(\theta_0,\varphi_0)}}\pi\lmk\caA_\ld\rmk''}^n
 \end{align}
 such that $T_\rho\pi=\rho$ and 
 $\sigma$-weak continuous on $\pi\lmk\caA_\ld\rmk''$ for any $\ld\in \ctvz$.
 Here $\overline\cdot^n$ means the norm-closure.
  \item[(iii)] The $C^*$-category $\col$ becomes a strict $C^*$-tensor category
  with tensor product
  \begin{align}
  \begin{split}
  \rho\otimes\sigma:= T_\rho T_\sigma\pi,\quad
  R\otimes S:= R T_\rho(S)
  \end{split}
  \end{align}
  for any $\rho,\sigma,\rho',\sigma'\in \ool$ 
  and $R\in \Mor_{\col}\lmk\rho,\rho'\rmk$
  and $S\in \Mor_{\col}\lmk\sigma,\sigma'\rmk$.
  The tensor unit is $\pi$.
  \item[(iv)]
  For any $\rho,\sigma\in \ool$,
  the norm limit
  \begin{align}
  \begin{split}
  \epsilon(\rho,\sigma):=\lim_{t\to\infty}
  V_{\sigma \ld_2(t)}T_\rho\lmk V_{\sigma \ld_2(t)}^*\rmk
  \end{split}
  \end{align}
  exists and is independent of the choice of $\ld_2\in \ctvz$ with $\ld_{0}\leftarrow_{(\theta_0,\varphi_0)} \ld_2$
  and $V_{\sigma \ld_2(t)}\in \caV_{\sigma \ld_2(t)}$.
  This $\epsilon(\rho,\sigma)$ gives a braiding of $\col$.
 \end{description}

\end{thm}
Note that
\begin{align}\label{zzz}
\begin{split}
Z\lmk \pi\lmk\caA_{\ld^c}\rmk'\cap\pi(\caA)''\rmk=Z\lmk\pi(\caA)''\rmk
\end{split}
\end{align}
for any cone $\ld$. This is because
$\pi(\caA_\ld)''\subset \pi(\caA_{\ld^c})'\cap \pi(\caA)''$ implies
\begin{align}
\begin{split}
Z\lmk \pi\lmk\caA_{\ld^c}\rmk'\cap\pi(\caA)''\rmk
\subset \pi\lmk\caA_{\ld^c}\rmk'\cap\pi(\caA)''\cap \pi(\caA_\ld)'
=Z\lmk\pi(\caA)''\rmk
\end{split}
\end{align}
and
\begin{align}
\begin{split}
&Z\lmk\pi(\caA)''\rmk
=\pi(\caA)''\cap \pi(\caA_{\ld^c})'\cap \pi\lmk\caA_\ld\rmk'
=\lmk  \pi(\caA)''\cap \pi(\caA_{\ld^c})'\rmk \cap \lmk \pi(\caA_{\ld^c})'\cap \pi\lmk\caA_\ld\rmk'\rmk\\
&=\lmk  \pi(\caA)''\cap \pi(\caA_{\ld^c})'\rmk \cap \pi(\caA)'
\subset 
\lmk \pi\lmk\caA_{\ld^c}\rmk'\cap\pi(\caA)''\rmk
\cap \lmk \pi\lmk\caA_{\ld^c}\rmk'\cap\pi(\caA)''\rmk'\\
&=Z\lmk \pi\lmk\caA_{\ld^c}\rmk'\cap\pi(\caA)''\rmk.
\end{split}
\end{align}
Therefore, by the approximate Haag duality, we have $Z\lmk\pi(\caA)''\rmk\subset \caF$.
In fact, for any $z\in Z\lmk\pi(\caA)''\rmk$ and any cone $\ld\in \ctvz$, with sufficiently small $\varepsilon>0$
(so that $\ld_\varepsilon \in \ctvz$), we have
\begin{align}
z=\Ad U_{\Lambda,\varepsilon}^* (z)\in \pi\lmk\caA_{\Lambda_{\varepsilon}
(-R_{\zeta,\varepsilon})}\rmk''\subset \caF,
\end{align}
where $\zeta:=\frac12{|\arg\ld|}$ (with the notation in Definition \ref{AHdef}).
Furthermore, from (\ref{sss}), the action of $T_\rho$ on $ \pi\lmk\caA_{\Lambda_{\varepsilon}(-R_{\zeta,\varepsilon})}\rmk''$
is equal to $\Ad\lmk V_{\rho,K}\rmk $ with any cone $K$ such that $K\cap {\Lambda_{\varepsilon}
(-R_{\zeta,\varepsilon})}
=\emptyset$.
Therefore, we obtain $T_\rho(z)=z\cdot \rho(\unit)$ for any $z\in Z\lmk\pi(\caA)''\rmk$.

The proof of Theorem \ref{same} is the same as that in \cite{MTC} and \cite{BA},
except for the slight difference in the existence of subobjects.
We give a proof only for this point:
\begin{lem}\label{shachi}
Let $\caA$ be a $2$-dimensional quantum spin system.
Let $\omega$ be a state on 
$\caA$ with properly infinite cone algebras satisfying the approximate Haag duality.
Let $(\caH,\pi)$ be a GNS representation of $\omega$.
Let $\theta_0\in \bbR$, $0<\varphi_0<\pi$ and $\ld_0\in \ctvz$.
Let $\rho\in \ool$ and $p\in \Mor_{\col}(\rho,\rho)$ be a non-zero projection.
Then there exists an object $\gamma\in \tool$ with 
$\gamma(\unit)=Z_{\pi(\caA)''}(p)$ and $v\in  \Mor_{\col}(\gamma,\rho)$ an isometry
such that $vv^*=p$.
\end{lem}
\begin{proof}
For each cone $\ld$, fix cones $D_\ld$ and $\Gamma_\ld$ so that
$D_\ld\cap\Gamma_\ld=\emptyset$ and $D_\ld,\Gamma_\ld\subset \ld$.
We also fix $V_{\rho,\Gamma_\ld}\in \caV_{\rho,\Gamma_\ld}$.

Note that
\begin{align}
p_{\Gamma_\ld}:=\Ad V_{\rho,\Gamma_\ld}^* \lmk p\rmk \in \pi\lmk\caA_{\Gamma_\ld^c}\rmk'\cap \pi\lmk\caA\rmk''
\subset \pi(\caA_{\ld^c})'\cap\pi(\caA)''
\end{align}
and $p_{\Gamma_\ld}$ is a projection.
We now show that $p_{\Gamma_\ld}$ is properly infinite in 
$\pi\lmk\caA_{\ld^c}\rmk'\cap\pi(\caA)''$.
Let $z\in Z\lmk \pi\lmk\caA_{\ld^c}\rmk'\cap\pi(\caA)''\rmk$
be a projection and assume that $zp_{\Gamma_\ld}$ is non-zero.
We have to show that $zp_{\Gamma_\ld}$ is infinite  in 
$\pi\lmk\caA_{\ld^c}\rmk'\cap\pi(\caA)''$.
Because $\pi\lmk D_\ld\rmk''$ is properly infinite, 
there exists a projection $E_\ld\in \pi\lmk D_\ld\rmk''$
and isometries $v_\ld,w_\ld\in \pi\lmk D_\ld\rmk''$
such that $v_\ld v_\ld^*=E_\ld$, $w_\ld w_\ld^*=\unit-E_\ld$.
%By Proposition 6.2.8 of \cite{KR}, the central carriers 
%$Z_{\pi\lmk D_\ld\rmk''}(E_\ld)$ , 
%$Z_{\pi\lmk D_\ld\rmk''}\lmk \unit-E_\ld\rmk$ of $E_\ld$ and $\unit-E_\ld$
%in $\pi\lmk D_\ld\rmk''$ are $\unit$.
Then $z p_{\Gamma_\ld} v_\ld\in \pi\lmk\caA_{\ld^c}\rmk'\cap\pi(\caA)''$ and
\begin{align}
\begin{split}
\lmk z p_{\Gamma_\ld} v_\ld\rmk\lmk z p_{\Gamma_\ld} v_\ld\rmk^*=z p_{\Gamma_\ld} E_\ld,\quad
\lmk z p_{\Gamma_\ld} v_\ld\rmk^* \lmk z p_{\Gamma_\ld} v_\ld\rmk=z p_{\Gamma_\ld}.
\end{split}
\end{align}
Because $\lmk z p_{\Gamma_\ld} w_\ld\rmk^* \lmk z p_{\Gamma_\ld} w_\ld\rmk=z p_{\Gamma_\ld}\neq 0$,
we have 
\begin{align}
\begin{split}
0\neq \lmk z p_{\Gamma_\ld} w_\ld\rmk\lmk z p_{\Gamma_\ld} w_\ld\rmk^*=z p_{\Gamma_\ld}\lmk \unit- E_\ld\rmk.
\end{split}
\end{align}
This proves that $zp_{\Gamma_\ld}$ is infinite  in 
$\pi\lmk\caA_{\ld^c}\rmk'\cap\pi(\caA)''$.
Hence $p_{\Gamma_\ld}$ is properly infinite in $\pi\lmk\caA_{\ld^c}\rmk'\cap\pi(\caA)''$.
Similarly, the central carrier
\begin{align}
\begin{split}
z_\ld:=Z_{\pi\lmk\caA_{\ld^c}\rmk'\cap\pi(\caA)''}\lmk p_{\Gamma_\ld}\rmk
\end{split}
\end{align}
is also properly infinite in $\pi\lmk\caA_{\ld^c}\rmk'\cap\pi(\caA)''$.
%In fact, for any projection $z\in Z\lmk \pi\lmk\caA_{\ld^c}\rmk' \cap \pi(\caA)''\rmk$,
%$z p_{\Gamma_\ld}= 0$ implies $z z_\ld =0$ 
%by the definition of the central carrier.
%Therefore, if $z z_\ld \neq 0$, then $z p_{\Gamma_\ld}$ is infinite in $\pi\lmk\caA_{\ld^c}\rmk'\cap\pi(\caA)''$.
%Namely, there exists $w\in \pi\lmk\caA_{\ld^c}\rmk'\cap\pi(\caA)''$ such that
%$ww^*\lneq z p_{\Gamma_\ld}$ and $w^*w=z p_{\Gamma_\ld}$.
%Because $w=p_{\Gamma_\ld} w p_{\Gamma_\ld}$, 
%\begin{align}
%W:=z z_\ld  (\unit- p_{\Gamma_\ld})+ z z_\ld   w\in \pi\lmk\caA_{\ld^c}\rmk'\cap\pi(\caA)''
%\end{align}
%satisfies
%\begin{align}
%\begin{split}
%W^*W=z z_\ld ,\quad W W^*= z z_\ld (\unit- p_{\Gamma_\ld})+z z_\ld ww^*
%\le z z_\ld.
%\end{split}
%\end{align}
%Note that we have
%\begin{align}
%z z_\ld-\lmk  z z_\ld (\unit- p_{\Gamma_\ld})+z z_\ld ww^*\rmk=
%z z_\ld\lmk p_{\Gamma_\ld}-ww^*\rmk
%=z p_{\Gamma_\ld}-ww^*\neq 0.
%\end{align}
%Hence for any $z\in Z\lmk \pi\lmk\caA_{\ld^c}\rmk'\cap\pi(\caA)''\rmk$, $zz_\ld$
%is infinite in $ \pi\lmk\caA_{\ld^c}\rmk'\cap\pi(\caA)''$ if it is non-zero.
%This proves that $z_\ld$ is 
%properly infinite in $\pi\lmk\caA_{\ld^c}\rmk'\cap\pi(\caA)''$.

Hence, for each cone $\ld$, both $p_{\Gamma_\ld}$ and $z_\ld$
are properly infinite in $\pi\lmk\caA_{\ld^c}\rmk'\cap\pi(\caA)''$
with central carrier $z_\ld$.
Therefore, by Corollary 6.3.5 of \cite{KR}, they are equivalent in $\pi\lmk\caA_{\ld^c}\rmk'\cap\pi(\caA)''$.
Namely, there exists $u_\ld\in \pi\lmk\caA_{\ld^c}\rmk'\cap\pi(\caA)''$
such that 
\begin{align}
u_\ld u_\ld^*=p_{\Gamma_\ld},\quad u_\ld^* u_\ld=z_\ld.
\end{align}

We claim $z_\ld$ is equal to the central carrier
$Z_{\pi(\caA)''}(p)$
 of $p$ in $\pi(\caA)''$, for any cone $\ld$.
In fact, because of (\ref{zzz}), $z_\ld$
is the intersections of projections $z$ in $Z\lmk \pi\lmk\caA_{\ld^c}\rmk'\cap\pi(\caA)''\rmk=Z\lmk\pi(\caA)''\rmk$
satisfying $z p_{\Gamma_\ld}=p_{\Gamma_\ld}$.
Note that for any projection $z\in Z\lmk \pi\lmk\caA_{\ld^c}\rmk'\cap\pi(\caA)''\rmk=Z\lmk\pi(\caA)''\rmk$,
 $z p_{\Gamma_\ld}=p_{\Gamma_\ld}$ if and only if 
$zp=p$ because
\begin{align}
\begin{split}
V_{\rho,\Gamma_\ld}
z p_{\Gamma_\ld} V_{\rho,\Gamma_\ld}^*
=
 z\rho(\unit) p\rho(\unit)
 =zp,\quad V_{\rho,\Gamma_\ld}
p_{\Gamma_\ld} V_{\rho,\Gamma_\ld}^*
 =p.
\end{split}
\end{align}
Hence, we have $z_\ld=Z_{\pi(\caA)''}(p)$, proving the claim.

Now, we define the object $\gamma$ by
\begin{align}
\gamma:=\Ad \lmk u_{\ld_0}^*V_{\rho\Gamma_{\ld_0}}^*\rmk\circ\rho
\end{align}
It is straightforward to show that this gives a $*$-representation of $\caA$.
 We can also check 
 \begin{align}
 \begin{split}
 V_{\gamma\ld}
 :=u_{\ld_0}^* V_{\rho\Gamma_{\ld_0}}^*
 V_{\rho\Gamma_{\ld}} u_{\ld}
 \in \caV_{\gamma\ld}
 \end{split}
 \end{align}
 for each $\ld$ and obtain $\gamma\in \ool$.
For 
\begin{align}
v:=V_{\rho\Gamma_{\ld_0}} u_{\ld_0},
\end{align}
we have
\begin{align}
\begin{split}
v^* v=u_{\ld_0}^* u_{\ld_0} =z_{\ld_0}=Z_{\pi(\caA)''}(p)=\gamma(\unit)=\id_{\gamma}, \quad
 v v^*= \rho(\unit)p\rho(\unit)=p
\end{split}
\end{align}
 and $v\in \Mor_{\col}(\gamma,\rho)$.

Note that  
\begin{align}
\begin{split}
\gamma(\unit)=V_{\gamma\ld}V_{\gamma\ld}^*=V_{\gamma\ld}^*V_{\gamma\ld}
 =Z_{\pi(\caA)''}(p).
\end{split}
\end{align}
Hence we have $V_{\gamma\ld}
\in\tilde 
 V_{\gamma\ld} $ for any cone $\ld$ and $\gamma\in \tool$. 
\end{proof}
In particular, setting $p=\rho(\unit)$ in Lemma \ref{shachi},
any $\rho\in \Obj \col$ is isomorphic to
some $\gamma\in\tool$.
\begin{prop}\label{sava}
Let $\caA$ be a two-dimensional quantum spin system.
Let $\omega$ be a state on $\caA$ with properly infinite cone algebras, satisfying the approximate Haag duality.
Let $\theta_0\in \bbR$, $0<\varphi_0<\pi$ and $\ld_0\in \ctvz$.
By setting objects
$
\Obj \tcol := \tool
$
 and morphisms between objects $\rho,\sigma\in \Obj \tcol $
 $
 \Mor_{\tcol}\lmk\rho,\sigma\rmk
 := \Mor_{\col}\lmk\rho,\sigma\rmk
 $,
 $\tcol$ becomes a braided $C^*$-tensor category with respect to the tensor
of $\col$.
Furthermore, $I_{\omega,\ld_0} : \tcol\to \col$ given by
\begin{align}
\begin{split}
&I_{\omega,\ld_0}(\rho):=\rho,\quad \rho\in \Obj\tcol,\\
& I_{\omega,\ld_0}(R):= R,\quad \rho,\sigma\in \Obj\tcol, \quad R\in  \Mor_{\tcol}\lmk\rho,\sigma\rmk
\end{split}
\end{align}
is an equivalence of braided $C^*$-tensor categories.
\end{prop}
\begin{proof}
That $\tcol$ is closed under tensor product, direct sum, subobjects can be checked
by using $V_{\rho\ld_0}\in {\tilde\caV}_{\rho\ld_0}$ (instead of general 
$V_{\rho\ld_0}\in {\caV}_{\rho\ld_0}$) in the constructions.
That $I_{\omega,\ld_0}$ is a fully faithful functor is trivial by the definition.
Lemma \ref{shachi} with $p=\rho(\unit)$ proves that $I_{\omega,\ld_0}$ is essentially surjective.
\end{proof}
The following Corollary says that the choice of the cone $\ld_0$ in Theorem \ref{same} does not matter.
\begin{cor}
Consider the setting in Theorem \ref{same}.
 If $\theta_0'\in\bbR$, $\varphi_0'\in (0,\pi)$
 $\ld_0'\in \caC_{\theta_0',\varphi_0'}$ is another choice of the cone, then $\col$ and $\coll{\ld_0'}$
are equivalent as braided $C^*$-tensor categories.
\end{cor}
\begin{proof}
By the same argument as \cite{MTC}, $\tcol $
and ${\tilde C_{\omega,\ld_0'}}$ are equivalent as braided $C^*$-tensor categories.
From Proposition \ref{sava},
this means  
$\col $
and ${C_{\omega,\ld_0'}}$ are equivalent as braided $C^*$-tensor categories.
\end{proof}

\begin{cor}
Consider the setting in Proposition \ref{sava}.
If, in addition, $\omega$ is a pure state, $\tcol$ is the same as the braided
$C^*$-tensor category $C^{\rm pure}_{\omega,\ld_0}$ given in \cite{MTC}.
In particular, $\col$ is equivalent to $C^{\rm pure}_{\omega,\ld_0}$ as
braided $C^*$-tensor categories.
\end{cor}
\begin{proof}
Recall objects in  $C^{\rm pure}_{\omega,\ld_0}$ are representations $\rho$ of $\caA$ on $\caH_\omega$
that have nonempty 
\begin{align}
\begin{split}
\caV_{\rho\ld}^{\rm pure}:=
\left\{  \Vrd\in \caU(\caH_\omega)\mid
\left. \Ad \Vrd \circ\pi_\omega\right\vert_{\caA_{\ld^c}}=\left. \rho\right\vert_{\caA_{\ld^c}}
\right\}\neq\emptyset,
\end{split}
\end{align}for any cone $\ld$, with $\unit\in \caV_{\rho\ld_0}^{\rm pure}$.

When $\omega$ is pure, $\pi_\omega(\caA)''=\caB(\caH_\omega)$, and it is a factor.
In this situation, we claim $\Obj \tcol=\Obj C^{\rm pure}_{\omega,\ld_0}$, and
$\caV_{\rho\ld}^{\rm pure}=\tvrd$ for any cone $\ld$.
 In fact for $\rho\in \Obj \tcol$, we have $\rho(\unit)\in Z\lmk\pi_\omega(\caA)''\rmk=\bbC\unit$, hence
 $\rho(\unit)=\unit$. As a result,  any element in $\tvrd$
has to be a unitary.
This means $\tvrd\subset \caV_{\rho\ld}^{\rm pure}$ and 
$\rho\in \Obj C^{\rm pure}_{\omega,\ld_0}$.
Hence $  \Obj \tcol\subset \Obj C^{\rm pure}_{\omega,\ld_0}$.
Conversely, if $\rho\in \Obj C^{\rm pure}_{\omega,\ld_0}$, then
$\rho(\unit)=\unit\in Z\lmk\pi_\omega\lmk\caA\rmk''\rmk$
and $\caV_{\rho\ld}^{\rm pure}\subset \tvrd$, 
hence
 we have $\Obj C^{\rm pure}_{\omega,\ld_0}\subset 
\Obj \tcol$.This proves the claim.

Because $\pi_\omega(\caA)''=\caB(\caH_\omega)$,
\begin{align}
\Mor_{C^{\rm pure}_{\omega,\ld_0}}(\rho,\sigma)
=\left\{R\in \caB(\caH_\omega)\mid
R\rho(A)=\sigma(A)R\; \text{for all}\; A\in \caA
\right\}
=\Mor_{\tcol}(\rho,\sigma)
\end{align} holds for any $\rho,\sigma\in \Obj \tcol=\Obj C^{\rm pure}_{\omega,\ld_0}$.
As the result of $\caV_{\rho\ld}^{\rm pure}=\tvrd$, tensor products, direct sums, subobjects and  braidings given
in terms of them are the same in $\tcol$ and $C^{\rm pure}_{\omega,\ld_0}$.
\end{proof}

\section{Braided $C^*$-tensor categories of subsystems}\label{subsys}
Let $\caA$, $\caB$ be two-dimensional quantum spin systems.
In this section, 
we consider the relation between the category of the system $\caA\otimes\caB$
and the category of the subsystem $\caA$.
First, we relate the GNS representations of a state $\omega$ on $\caA\otimes\caB$ and its restriction to the subsystem ${\omega\vert_{\caA}}$.
\begin{lem}\label{hina}
Let $\caA$, $\caB$ be two-dimensional quantum spin systems.
Let $\omega$ be a state on $\caA\otimes \caB$ with a GNS triple $(\caH,\pi,\Omega)$.
Let $\varphi:=\omega\vert_{\caA}$ be the restriction of
$\omega$ onto $\caA$.
Let $p$ be the orthogonal projection
onto the closed subspace $\overline{\pi(\caA)\Omega}$ in $\caH$.
Then the following hold.
\begin{description}
\item[(i)]
The triple
$(p\caH,\pi_{\varphi},\Omega)$
with $\pi_\varphi(A):=\pi(A) p$, $A\in\caA$
is a GNS triple of $\varphi$,
\item[(ii)]The map
\begin{align}
\Theta :\pi(\caA)''\ni x\mapsto xp\in  \pi_\varphi(\caA)''
\end{align}
is a $*$-isomorphism satisfying
\begin{align}\label{goldenring}
\Theta\lmk \pi(A)\rmk=\pi_\varphi(A),\quad A\in \caA.
\end{align}
\end{description}
\end{lem}
\begin{proof}
(i) is trivial from the setting.
The map $\Theta$ in (ii) is clearly a surjective $*$-homomorphism because $p\in \pi(\caA)'$.
If $\Theta(x)=0$ for some $x\in \pi(\caA)''$, then
\begin{align}
\begin{split}
x\pi(B)\pi(A)\Omega=\pi(B) x \pi(A)\Omega
=\pi(B) x p\pi(A)\Omega=0,
\end{split}
\end{align}
for any $A\in\caA$ and $B\in {\caB}$.
Because $\Omega$ is cyclic for $\pi\lmk \caA\otimes{\caB}\rmk$, we have
$x=0$.
Hence $\Theta$ is injective hence an $*$-isomorphism.
\end{proof}
\begin{thm}\label{cofe}
Let $\caA$, $\caB$ be two-dimensional quantum spin systems.
Let $\omega$ be a state on $\caA\otimes \caB$
and $\varphi:=\omega\vert_{\caA}$ the restriction of
$\omega$ onto $\caA$.
Suppose that both $\omega$ and $\varphi$ have properly infinite cone algebras
and satisfy the approximate Haag duality.
Let $(\caH,\pi)$, $(\caH_\varphi, \pi_\varphi)$ be GNS representations of $\omega$
and $\varphi$ and let  $\Theta : \pi(\caA)''\to \pi_\varphi(\caA)''$ be the $*$-isomorphism satisfying (\ref{goldenring}) (see  Lemma \ref{hina}).
Let $\theta_0\in\bbR$, $0<\varphi_0<\pi$ and $\ld_0\in \ctvz$.
Then
\begin{align}
\folz(\rho):=\hat \rho,\quad \rho\in \Obj \cols{\varphi}
\end{align}
with
\begin{align}\label{rhat}
\begin{split}
\hat\rho\lmk A\otimes B\rmk
=\Theta^{-1}\lmk\rho(A)\rmk \pi(B),\quad
A\in \caA,\quad B\in \caB
\end{split}
\end{align}
and 
\begin{align}\label{matrix}
\begin{split}
\folz(R):=\Theta^{-1}(R),\quad
\rho,\sigma\in\Obj \cols{\varphi},\quad R\in \Mor_{\cols{\varphi}}\lmk \rho,\sigma\rmk
\end{split}
\end{align}
defines a strict faithful braided tensor functor 
$\folz : \cols{\varphi}\to \col$.
\end{thm}
\begin{proof}
Note from Remark \ref{panda} that for each $\rho\in\Obj \cols{\varphi}$,
$\Theta^{-1}\rho :\caA\to{\caB}(\caH)$ is a well-defined representation of
$\caA$ on $\caH$,
whose range commutes with that of $\pi\vert_{{\caB}}$.
Therefore, there exists a unique representation
$\hat \rho$ of $\caA\otimes{\caB}$
satisfying (\ref{rhat}). (See Proposition 3.3.7 of \cite{brown2008textrm}.)

We claim $\Theta^{-1}\lmk \vrd\rmk\subset\caV_{\hat\rho\ld}$
for any $\rho\in \Obj\cols\varphi$ and a cone $\ld$.
In fact, for any $\Vrd\in \vrd$, because $\Theta$ is a map from $\pi(\caA)''$ 
to  $\pi_\varphi(\caA)''$,
we have $\Theta^{-1}(\Vrd)\in \pi(\caA)''$.
Using this fact and (\ref{goldenring}), we have
\begin{align}
\begin{split}
&\Ad\lmk \Theta^{-1}\lmk V_{\rho\Lambda}\rmk\rmk
\lmk
\pi(A\otimes B)
\rmk\\
&=\lmk
\Ad\lmk\Theta^{-1}\lmk V_{\rho\Lambda}\rmk\rmk
\lmk
\pi(A)
\rmk\rmk\cdot
\pi(B)
\\
&=\lmk
\Ad\lmk \Theta^{-1}\lmk V_{\rho\Lambda}\rmk\rmk
\lmk
\Theta^{-1}\pi_\varphi(A)
\rmk\rmk\cdot
\pi(B)
\\
&=
\lmk \Theta^{-1}\Ad \lmk V_{\rho\Lambda}\rmk\pi_\varphi(A)\rmk\cdot \pi(B)\\
&=\lmk \Theta^{-1}\rho(A)\rmk\cdot \pi(B)
\\
&=\hat\rho(A\otimes B)
\end{split}
\end{align}
for any $A\in\caA_{\ld^c}$ and $B\in{\caB_{\ld^c}}$.
We also have
\begin{align}
\begin{split}
&\Theta^{-1}\lmk V_{\rho\Lambda}\rmk^* \Theta^{-1}\lmk V_{\rho\Lambda}\rmk
=\Theta^{-1}\lmk  V_{\rho\Lambda}^*  V_{\rho\Lambda}  \rmk\\
&\in \pi(\caA_{\ld^c})'\cap \pi(\caB)'
\subset
 \pi\lmk\lmk \caA\otimes\caB\rmk_{\ld^c}\rmk'.
\end{split}
\end{align}
Hence we have $\Theta^{-1}(\Vrd)\in \caV_{\hat\rho\ld}$ and $\hat\rho\in \caO_\omega$.
Furthermore, for any $A\in \caA_{\Lambda_0^c}$ and
$B\in {\caB}_{\Lambda_0^c}$, we have
\begin{align}
\begin{split}
\hat\rho\lmk A\otimes B\rmk
=\Theta^{-1}\lmk\rho(A)\rmk \pi\lmk B\rmk
=\Theta^{-1}\lmk \pi_\varphi(A) \rho(\unit)\rmk  \pi\lmk B\rmk
= \pi\lmk A\otimes B\rmk \Theta^{-1}\lmk \rho(\unit)\rmk
=\pi\lmk A\otimes B\rmk \hat\rho(\unit).
\end{split}
\end{align}
Hence we conclude $\hat\rho\in \ool$.

Similarly, it is straightforward to check that (\ref{matrix})
defines a morphism in $\col$
from $\hat\rho$ to $\hat\sigma$
for each $\rho,\sigma\in\Obj \cols{\varphi}$,
$R\in \Mor_{\cols{\varphi}}\lmk \rho,\sigma\rmk$.
In fact, $\Theta^{-1}(R)\in \pi(\caA)''\subset\pi(\caA\otimes\caB)''$ is well defined  by the definition of
$\Theta$.
For $A\in \caA$ and $B\in \caB$, we have
\begin{align}
\begin{split}
&\Theta^{-1}(R) \hat\rho(A\otimes B)
=\Theta^{-1}(R) \Theta^{-1}\lmk \rho(A)\rmk \pi(B)
=\Theta^{-1}\lmk R \rho(A)\rmk \pi(B)
=\Theta^{-1}\lmk \sigma(A) R \rmk \pi(B)\\
&=\Theta^{-1}\lmk \sigma(A) \rmk  \Theta^{-1}\lmk  R \rmk \pi(B)
=\Theta^{-1}\lmk \sigma(A) \rmk \pi(B)  \Theta^{-1}\lmk  R \rmk 
=\hat\sigma(A\otimes B) \Theta^{-1}\lmk  R \rmk 
\end{split}
\end{align}
because $\Theta^{-1}(R)\in  \pi(\caA)''$ commutes with $\pi(\caB)$.
Furthermore, 
\begin{align}
\begin{split}
\Theta^{-1}\lmk\sigma(\unit)\rmk\Theta^{-1}(R) \Theta^{-1}\lmk\rho(\unit)\rmk
=\Theta^{-1}\lmk\sigma(\unit)\cdot R\cdot \rho(\unit)\rmk
=\Theta^{-1}(R).
\end{split}
\end{align}
Hence we have 
\begin{align}
\begin{split}
\folz(R)=\Theta^{-1}(R)
\in \Mor_{\col}\lmk\hat\rho,\hat\sigma\rmk.
\end{split}
\end{align}

Obviously we have $\folz(SR) =\folz(S)\folz(R)$
when $S$ and $R$ are composable, and 
\begin{align}
\folz\lmk \id_\rho\rmk=\Theta^{-1}(\rho(\unit))=\hat\rho(\unit)
=\id_{\folz(\rho)}.
\end{align}
Hence $\folz$ is a functor.Clearly it is faithful.

Next we show that $\folz$ is a strict tensor functor.
In fact, $\folz(\pi_\varphi)=\pi$ because
\begin{align}
\folz(\pi_\varphi)\lmk A\otimes B\rmk
=\Theta^{-1}\lmk \pi_\varphi(A)\rmk \cdot \pi(B)
=\pi\lmk A\otimes B\rmk,
\end{align}
for any $A\in\caA$ and $B\in\caB$.

To see that $\hat\rho\otimes\hat\sigma=\folz(\rho)\otimes \folz(\sigma)$
is equal to $\widehat{\rho\otimes\sigma}=\folz(\rho\otimes\sigma)$,
first we claim 
that 
\begin{align}\label{koara}
\begin{split}
T_{\hat\rho}(x)=\Theta^{-1}T_\rho\Theta(x),\quad
\text{for all}\quad x\in \overline{\cup_{\ld \in \ctvz}\pi(\caA_\ld)''}^{n}
\end{split}
\end{align}
for $\rho\in \ool$. (Recall (ii) if Theorem \ref{same}.)
In fact we have
\begin{align}\label{kangal}
\begin{split}
\Theta^{-1}T_\rho\Theta\pi(A)
=\Theta^{-1}T_\rho\lmk \pi_\varphi(A)\rmk
=\Theta^{-1}\rho(A)
=\hat\rho(A)=T_{\hat\rho}\pi(A),
\end{split}
\end{align}
for any $A\in \caA$.
Because $\Theta\lmk \pi\lmk\caA_{\ld}\rmk''\rmk=\pi_\varphi(\caA_\ld)''$ and
$T_\rho$ is $\sigma$-weak continuous on $\pi_\varphi(\caA_\ld)''$
for any $\ld\in\ctvz$,
$\Theta^{-1}T_\rho\Theta$ is $\sigma$-weak continuous on $\pi(\caA_\ld)''$.
As $T_{\hat\rho}$ is also $\sigma$-weak continuous on $\pi(\caA_\ld)''$, combining 
with (\ref{kangal}), we conclude that $\Theta^{-1}T_\rho\Theta(x)=T_{\hat\rho}(x)$
for any $x\in \overline{\cup_{\ld \in \ctvz}\pi(\caA_\ld)''}^{n}$.
On the other hand,
$T_{\hat\rho}(y)=y\hat \rho(\unit)$ for any
$y\in \overline{\cup_{\ld \in \ctvz}\pi(\caB_\ld)''}^{n}$,
because $\hat\rho(B)=\pi(B)\hat \rho(\unit)$ for any $B\in \caB$.
Using these facts, for any $A\in\caA$ and $B\in\caB$, we have
\begin{align}
\begin{split}
&\lmk \hat\rho\otimes \hat\sigma\rmk(A\otimes B)=
T_{\hat\rho}T_{\hat\sigma}\pi(A\otimes B)
=T_{\hat\rho}T_{\hat\sigma}\pi(A)\cdot T_{\hat\rho}T_{\hat\sigma}\pi(B)\\
&=\Theta^{-1}T_\rho\Theta\Theta^{-1}T_\sigma\Theta\pi(A)\cdot
\pi(B)
=\Theta^{-1}T_\rho T_\sigma\pi_\varphi(A)\cdot
\pi(B)\\
&=\Theta^{-1}\rho\otimes \sigma(A)\cdot
\pi(B)
=\widehat{\rho\otimes\sigma}\lmk A\otimes B\rmk.
\end{split}
\end{align}
Hence we have
$\hat\rho\otimes\hat\sigma=\widehat{\rho\otimes\sigma}$.

To see the identity morphisms give natural isomorphisms,
let $\rho,\rho',\sigma,\sigma'\in\Obj \cols{\varphi}$
and
$R\in \Mor_{\cols{\varphi}}(\rho,\rho')$, 
$S\in \Mor_{\cols{\varphi}}(\sigma,\sigma')$.
Recalling (\ref{koara}) and
$\Theta^{-1}(S)\in \overline{\cup_{\ld \in \ctvz}\pi(\caA_\ld)''}^{n}$, we have
\begin{align}
\begin{split}
&\folz(R)\otimes \folz(S)
=\Theta^{-1}(R) T_{\hat \rho} \lmk \Theta^{-1}(S)\rmk
=\Theta^{-1}(R)
\Theta^{-1}T_\rho\Theta\Theta^{-1}(S)\\
&=\Theta^{-1}\lmk R T_\rho(S)\rmk
=\Theta^{-1}\lmk R \otimes S\rmk
=\folz\lmk  R \otimes S\rmk.
\end{split}
\end{align}
Hence identity morphisms give natural isomorphisms
$\folz(\rho)\otimes \folz(\sigma)\to \folz(\rho\otimes\sigma)$.
Hence $\folz: \cols{\varphi}\to\col$ is a strict tensor functor.

To see that $\folz$ is braided,
let $\rho,\sigma\in\Obj\col$,
 $\ld_2\in \ctvz$ with $\ld_{0}\leftarrow_{(\theta_0,\varphi_0)} \ld_2$
  and $V_{\sigma \ld_2(t)}\in \caV_{\sigma \ld_2(t)}$.
From the above observation, we know that $\Theta^{-1}\lmk V_{\sigma \ld_2(t)} \rmk
\in \caV_{\hat \sigma \ld_2(t)}$.
From this and (\ref{koara}), we have
\begin{align}
\begin{split}
&\epsilon\lmk\hat\rho,\hat\sigma \rmk
=\lim_{t} \Theta^{-1}\lmk V_{\sigma \ld_2(t)} \rmk
T_{\hat\rho}
\lmk
\Theta^{-1}\lmk V_{\sigma \ld_2(t)}^* \rmk
\rmk
=
\lim_{t} \Theta^{-1}\lmk V_{\sigma \ld_2(t)} \rmk
\Theta^{-1}T_{\rho}\Theta
\lmk
\Theta^{-1}\lmk V_{\sigma \ld_2(t)}^* \rmk
\rmk\\
&=\lim_t 
 \Theta^{-1}\lmk V_{\sigma \ld_2(t)} 
 T_{\rho}\lmk V_{\sigma \ld_2(t)}^* \rmk
 \rmk
 =\Theta^{-1}\lmk\epsilon(\rho,\sigma)\rmk
 =\folz\lmk\epsilon(\rho,\sigma)\rmk.
\end{split}
\end{align}
This proves that $\folz$ is braided.
\end{proof}
In some situations that we will consider, this functor is fully faithful.
\begin{lem}\label{washi}
In the setting of Theorem \ref{cofe},  
suppose that 
$\psi:=\omega\vert_{\caB}$ is pure.
Then, 
the tensor functor $\folz$ of  Theorem \ref{cofe}
is fully faithful.
\end{lem}
\begin{proof}
Let $(\caH_\psi,\pi_\psi,\Omega_\psi)$ and
$(\caH_\varphi,\pi_\varphi,\Omega_{\varphi})$ be GNS triples of $\psi$ and $\varphi:=\omega\vert_\caA$ respectively.
By the assumption, $\omega$ is of the tensor product form
(IV Lemma 4.11\cite{takesaki}), and 
we can take a GNS representation  $(\caH,\pi,\Omega)$
of $\omega$ of the form $\caH:=\caH_\varphi\otimes\caH_{\psi}$
,$\pi:= \pi_\varphi\otimes\pi_\psi$, $\Omega:=\Omega_\varphi\otimes\Omega_\psi$.
In this representation, the isomorphism of Lemma \ref{hina} is
$\Theta^{-1}(x)=x\otimes\unit$, $x\in\pi_\varphi(\caA)''$.
Because $\psi$ is pure, we have $\pi_\psi({\caB})''=\caB(\caH_{\psi})$.
Using this fact,
for any $\rho,\sigma\in \Obj \cols{\varphi}$, we have
\begin{align}
\begin{split}
&\Mor_{\col}
\lmk
\folz(\rho), \folz(\sigma)
\rmk\\
&=
\left\{
S\in \lmk \pi_\varphi\lmk \caA\rmk\rmk''
\bar\otimes \lmk \pi_\psi\lmk{\caB}\rmk\rmk''\middle |
\begin{gathered}
S \Theta^{-1}\lmk\rho(A)\rmk\pi(B) 
=\Theta^{-1}\lmk\sigma(A)\rmk\pi(B) S,\quad A\in \caA,\quad B\in{\caB} \\
\Theta^{-1}(\sigma(\unit)) S\Theta^{-1}(\rho(\unit)) =S
\end{gathered}
\right\}\\
&=
\left\{
S\otimes\unit_\psi
\middle |
\begin{gathered}
S\in \lmk \pi_\varphi\lmk \caA\rmk\rmk'',
S\rho(A)=\sigma(A)S,\quad A\in \caA\\
\sigma(\unit) S \rho(\unit)=S
\end{gathered}
\right\}\\
&=\Mor_{\cols\varphi}(\rho,\sigma)\otimes\unit_\psi.
\end{split}
\end{align}
This means $\folz$ 
is fully faithful.
\end{proof}

\section{Stabilization}\label{stabilization}
In this section, we show Theorem \ref{stab}.
\subsection{Preparation}
In this subsection, we provide Lemmas, which will be needed for the proof of Theorem \ref{stab}.
\begin{lem}\label{orange}
Let $\caA$, $\caB_1$, $\caB_2$ be infinite-dimensional UHF algebras
and set $\mathfrak A:=\caA\otimes \caB_1\otimes\caB_2$.
Let $\caA_1$ be a unital $C^*$-subalgebra of $\caA$.
Let $\caH_0$, $\caK_1$, $\caK_2$ be separable infinite-dimensional Hilbert spaces.
Let $\pi_0$ a representation of $\caA$ on $\caH_0$,
and $\pi_1$, $\pi_2$ irreducible representations of $\caB_1$, $\caB_2$
on $\caK_1$, $\caK_2$ respectively.
Set $\pi:=\pi_0\otimes\pi_1\otimes\pi_2$, the representation of $\mathfrak A$
on $\caH_0\otimes\caK_1\otimes\caK_2$.
Let $U$ be an element in $ \pi(\mathfrak A)''$ and $\rho$ a representation of $\mathfrak A$
on $\caH_0\otimes\caK_1\otimes\caK_2$,
and $z_\rho\in Z\lmk \pi_0(\caA)''\rmk\subset\caB(\caH_0)$
such that
\begin{align}
&\Ad U\circ\pi(A_1\otimes B_1)=\rho(A_1\otimes B_1),\quad A_1\in\caA_1,\quad  B_1\in\caB_1,\label{kiiro}\\
&U^*U\in \pi(\caA_1\otimes \caB_1)',\label{ao}\\
&\rho\lmk B_2\rmk=z_\rho\otimes\unit_{\caK_1}\otimes \pi_2(B_2),\quad
B_2\in \caB_2.\label{aka}
\end{align}
Then for any projection $p\in \pi(\caA\otimes\caB_1)''\cap\rho(\caA_1\otimes \caB_1)'$,
there exists a partial isometry $V\in  \pi(\caA\otimes\caB_1)''$
such that 
\begin{align}
&VV^*\le p,\quad q_0\otimes\unit_{\caK_1}\otimes\unit_{\caK_2}:=V^*V
\in \pi(\caA)''\cap\pi(\caA_1)',\label{mizuiro}\\
&V \pi(A_1\otimes B_1)=\rho(A_1\otimes B_1) V,\quad\text{for all}\quad A_1\in\caA_1,\quad B_1\in \caB_1,
\label{vprv}\\
&(p-VV^*) U(\unit-V^*V)=0,\label{kuro}\\
&p\cdot\lmk \lmk\unit-Z_{\pi_0(\caA)''}(q_0)\rmk z_\rho\otimes \unit_{\caK_1}\otimes \unit_{\caK_2}\rmk=0.
\label{kimidori}
\end{align}
\end{lem}
\begin{proof}
The argument is standard (see \cite{d1983interpolation}).
For the reader's convenience, we provide proof here.
Let us consider the system $\caS$ of sets 
of nonzero partial isometries $\{v_\lambda\}_{\lambda\in\ld}$
in $\pi(\caA\otimes\caB_1)''$
satisfying the following conditions:
\begin{description}
\item[(1)]
\begin{align}
\begin{split}
v_\lambda\pi(A_1\otimes B_1)=\rho(A_1\otimes B_1)v_\lambda,\quad \text{for all}\quad
A_1\in\caA_1,\quad
B_1\in \caB_1,\quad \lambda\in\ld.
\end{split}
\end{align}
\item[(2)] projections
$\{v_\lambda^*v_\lambda\}_{\lambda\in\ld}$ are mutually orthogonal,
\item[(3)]
projections
$\{v_\lambda v_\lambda^*\}_{\lambda\in\ld}$ are mutually orthogonal,
\item[(4)]
$v_\lambda v_\lambda^*\le p$
for any $\lambda\in\ld$. 
\end{description}
Note that $\caS$ is inductively ordered with respect to the inclusion order.
By Zorn's Lemma, there exists a maximal element
of $\caS$.
We fix a maximal $S=\{v_\lambda\}_{\lambda\in\Lambda}\in \caS$, and set
\begin{align}\label{toumei}
\begin{split}
V:=\sum_{\lambda\in\ld}v_\lambda
\in \pi(\caA\otimes\caB_1)''
\end{split}
\end{align}
which converges in the strong$*$-topology because of the condition (2), (3) above.
This $V$ satisfies the required properties.
We now check these one by one.

By definition, $V$ is a partial isometry in $\pi(\caA\otimes\caB_1)''$
satisfying $VV^*\le p$ and
\begin{align}\label{mochi}
\begin{split}
V\pi(A_1\otimes B_1) =\sum_{\lambda\in\ld}v_\lambda\pi(A_1\otimes B_1)
=\sum_{\lambda\in\ld}\rho(A_1\otimes B_1)v_\lambda
=\rho(A_1\otimes B_1) V,\quad A_1\in\caA_1,\quad B_1\in\caB_1.
\end{split}
\end{align}
Because of this property, we have
\begin{align}
\begin{split}
V^* V\in& \pi(\caA_1\otimes \caB_1)'\cap \pi(\caA\otimes\caB_1)''
=\lmk \pi_0(\caA_1)''\bar \otimes \pi_1(\caB_1)''\bar\otimes \bbC\unit_{\caK_2}\rmk'
\cap \lmk
\pi_0(\caA)''\bar\otimes \pi_1(\caB_1)''\bar\otimes \bbC\unit_{\caK_2}\rmk\\
&=\lmk \pi_0(\caA)''\cap\pi_0(\caA_1)'\rmk\bar\otimes \bbC\unit_{\caK_1}\bar\otimes \bbC\unit_{\caK_2}
=\pi(\caA)''\cap \pi(\caA_1)'.
\end{split}
\end{align}
Here, we used the irreducibility of $\pi_1$.
Hence we have proven  (\ref{mizuiro}) and (\ref{vprv}).

Next we show $p_1 U q=0$, corresponding to (\ref{kuro}),
where
\begin{align}
&p_1:=p-VV^*=p-\sum_\lambda v_\lambda v_\lambda^*,\\
&q:=\unit-V^*V
=\unit-\sum_\lambda v_\lambda^* v_\lambda,
\end{align}
are projections.
Because of (\ref{mochi}) and $p\in\rho(\caA_1\otimes \caB_1)'\cap \pi(\caA\otimes \caB_1)''$, we have $p_1\in \rho(\caA_1\otimes\caB_1)'\cap \pi(\caA\otimes \caB_1)''$ 
and $q\in\pi(\caA_1\otimes\caB_1)'\cap \pi(\caA\otimes \caB_1)''$.

We claim $\rho(\caA_1\otimes\caB_1)\subset \pi(\caA\otimes\caB_1)''$.
In fact, because of $U\in\pi(\mathfrak A)''$ and (\ref{kiiro}), (\ref{aka}), we have
\begin{align}\label{shiro}
\begin{split}
\rho(\caA_1\otimes\caB_1)&\subset \rho(\caB_2)'\cap 
\lmk \pi(\mathfrak A)''\lmk z_\rho\otimes \unit_{\caK_1}\otimes \unit_{\caK_2}\rmk\rmk\\
&=\lmk z_\rho\otimes\unit_{\caK_1}\otimes \pi_2(\caB_2)\rmk'\cap 
\lmk \pi(\mathfrak A)''\lmk z_\rho\otimes \unit_{\caK_1}\otimes \unit_{\caK_2}\rmk\rmk\\
%&=
%\lmk z_\rho\otimes\unit_{\caK_1}\otimes \pi_2(\caB_2)\rmk'\cap 
%\lmk \pi_0(\mathfrak A)'' z_\rho
%\bar\otimes \pi_1(\caB_1)''
%\bar \otimes \pi_2(\caB_2)''\rmk\\
&=
\lmk \pi_0(\caA)'' z_\rho
\bar\otimes \pi_1(\caB_1)''
\bar \otimes \bbC\unit_{\caK_2}\rmk
\subset \pi(\caA\otimes\caB_1)''.
\end{split}
\end{align}
Here we used the irreducibility of $\pi_2$.

Note that for any
nonzero $x\in \pi(\mathfrak A)''$, there exists
a $\sigma$-weak continuous projection of norm one
$\bbE^{(x)}$
from
$\pi(\mathfrak A)''$ onto $\pi(\caA\otimes \caB_1)''$
such that $\bbE^{(x)}(x)\neq 0$.
This is true because
\begin{align}
\begin{split}
&\pi(\mathfrak A)''=\pi_0(\caA)''\bar\otimes \pi_1(\caB_1)''\bar\otimes\pi_2(\caB_2)''
=\pi_0(\caA)''\bar\otimes \pi_1(\caB_1)''\otimes \caB(\caK_2),\\
&\pi(\caA\otimes\caB_1)''
=\pi_0(\caA)''\bar\otimes \pi_1(\caB_1)''\bar \otimes\bbC\unit_{\caK_2}.
\end{split}
\end{align}

Now assume that $x:=p_1 U q\neq 0$. We derive a contradiction out of this, proving (\ref{kuro}).
Because of $q\in \pi(\caA_1\otimes \caB_1)'$, $p_1\in \rho(\caA_1\otimes \caB_1)'$, and (\ref{kiiro})
for any $A_1\in\caA_1$,
$B_1\in\caB_1$, we have
\begin{align}\label{haiiro}
\begin{split}
&x\pi(A_1\otimes B_1)=p_1 U q \pi(A_1\otimes B_1)
=p_1 U\pi(A_1\otimes B_1) q=p_1 \rho(A_1\otimes B_1) U q\\
&=\rho(A_1\otimes B_1) p_1U q=\rho(A_1\otimes B_1) x.
\end{split}
\end{align}
Note that $x\in\pi(\mathfrak A)''$.
Because we assumed $x\neq 0$, there 
exists
a $\sigma$-weak continuous projection of norm one
$\bbE^{(x)}$
from
$\pi(\mathfrak A)''$ onto $\pi(\caA\otimes \caB_1)''$
such that $\pi(\caA\otimes \caB_1)''\ni t:=\bbE^{(x)}(x)\neq 0$.
Using the property of projection of norm one \cite{takesaki},
(\ref{shiro}), and (\ref{haiiro}),
for any $A_1\in\caA_1$, $B_1\in \caB_1$, 
we obtain 
\begin{align}
\begin{split}
&t \pi(A_1\otimes B_1)
=\bbE^{(x)}(x)\pi(A_1\otimes B_1) =\bbE^{(x)}\lmk x \pi(A_1\otimes B_1)\rmk
=\bbE^{(x)}\lmk \rho(A_1\otimes B_1) x \rmk\\
&= \rho(A_1\otimes B_1) \bbE^{(x)}\lmk x \rmk
=\rho(A_1\otimes B_1) t.
\end{split}
\end{align}
Taking the polar decomposition 
$t=v|t|$, we obtain a nonzero partial isometry
$v\in \pi(\caA\otimes\caB_1)''$
satisfying
\begin{align}
\begin{split}
v\pi(A_1\otimes B_1)=\rho(A_1\otimes B_1) v,\quad\text{for all}\quad
A_1\in\caA_1,\quad  B_1\in\caB_1.
\end{split}
\end{align}
By the definition of $t$, we have
\begin{align}
\begin{split}
t=v|t|=\bbE^{(x)}(x)=\bbE^{(x)}\lmk p_1 U q\rmk
=p_1 \bbE^{(x)}(x) q.
\end{split}
\end{align}
Here, we used $p,q\in \pi(\caA\otimes\caB_1)''$.
Therefore, we have
\begin{align}
vv^*\le p_1=p-\sum_{\lambda} v_\lambda v_\lambda^*\le p,\quad
v^* v\le q=\unit-\sum_{\lambda}  v_\lambda^* v_\lambda.
\end{align}
From this, we see that 
$\{v_\lambda\}_\lambda\cup\{v \}\in\caS$ and it is strictly larger than
$S=\{v_\lambda\}_\lambda$. 
This contradicts the maximality of $S$.
Hence, we conclude $p_1 U q=0$, proving (\ref{kuro}).

The last property (\ref{kimidori}) follows from  (\ref{kuro}).
To see this, first note that
\begin{align}\label{midori}
\begin{split}
Z_{\pi_0(\caA)''}(q_0)\otimes\unit_{\caK_1}\otimes \unit_{\caK_2}
\ge q_0 \otimes\unit_{\caK_1}\otimes \unit_{\caK_2}=V^*V
\end{split}
\end{align}
for the projection $q_0$ given in (\ref{mizuiro}).
From this and (\ref{kuro}), we obtain
\begin{align}
\begin{split}
&\lmk p-VV^*\rmk U \lmk\unit-Z_{\pi_0(\caA)''}(q_0) \otimes\unit_{\caK_1}\otimes \unit_{\caK_2}
\rmk\\
&=
\lmk p-VV^*\rmk U 
\lmk \unit- V^*V\rmk
\lmk\unit-Z_{\pi_0(\caA)''}(q_0)\otimes\unit_{\caK_1}\otimes \unit_{\caK_2}
\rmk=0.
\end{split}
\end{align}
We also have
\begin{align}
\begin{split}
VV^* U \lmk \unit-Z_{\pi_0(\caA)''}(q_0)\otimes\unit_{\caK_1}\otimes \unit_{\caK_2}
\rmk
=V\lmk \unit-Z_{\pi_0(\caA)''}(q_0)\otimes\unit_{\caK_1}\otimes \unit_{\caK_2}\rmk
V^* U 
=0
\end{split}
\end{align}
because $V^* U$ belongs to $\pi(\mathfrak A)''$,
and by (\ref{midori}).
From these equations, we obtain
\begin{align}
\begin{split}
 p\lmk\unit-Z_{\pi_0(\caA)''}(q_0) \otimes\unit_{\caK_1}\otimes \unit_{\caK_2}
\rmk U
=0.
\end{split}
\end{align}
Multiplying by $U^*$ from the right of this, we obtain
\begin{align}
\begin{split}
 p\lmk\lmk \unit-Z_{\pi_0(\caA)''}(q_0)\rmk z_\rho \otimes\unit_{\caK_1}\otimes \unit_{\caK_2}
\rmk 
=0,
\end{split}
\end{align}
because $UU^*=\rho(\unit)=z_\rho\otimes\unit_{\caK_1}\otimes\unit_{\caK_2}$.
\end{proof}
We use this Lemma to prove the following.
\begin{lem}\label{yamabuki}
Consider the setting of Lemma \ref{orange}.
Then there exists a sequence of 
partial isometries
$\{ V_n\}_{n=0}^N$ (with $N$ finite or infinite)
in $\pi(\caA\otimes\caB_1)''$
such that
\begin{align}
&V_n\pi(A_1\otimes B_1)=\rho(A_1\otimes B_1) V_n,\quad \text{for all}\quad
A_1\in\caA_1 \quad B_1\in \caB_1,\quad
\text{and}\quad n=0,\ldots,N\label{daidai}\\
&\sum_{n=0}^N V_n V_n^*=
z_\rho\otimes \unit_{\caK_1}\otimes \unit_{\caK_2},\label{pink}\\
&Z_{\pi(\caA)''}(V_0^*V_0)
=z_\rho\otimes\unit_{\caK_1}\otimes \unit_{\caK_2}\label{murasaki}.
\end{align}
\end{lem}
\begin{proof}
Applying Lemma \ref{orange} to 
$p=z_\rho\otimes \unit_{\caK_1}\otimes \unit_{\caK_2}
\in \pi\lmk\caA\otimes\caB_1\rmk''\cap\rho\lmk\caA_1\otimes \caB_1\rmk'$,
we obtain a partial isometry
$V_0\in \pi(\caA\otimes\caB_1)''$ such that
\begin{align}\label{chairou}
\begin{split}
&V_0V_0^*\le z_\rho\otimes \unit_{\caK_1}\otimes \unit_{\caK_2},\quad
V_0^* V_0\in \lmk \pi_0(\caA)''\cap\pi_0(\caA_1)'\rmk\bar\otimes\bbC\unit_{\caK_1}\bar \otimes\bbC\unit_{\caK_2}\\
&V_0 \pi(A_1\otimes B_1)=\rho(A_1\otimes B_1) V_0,\quad\text{for all}\quad A_1\in\caA_1,\quad  B_1\in \caB_1,
\\
%&(z_\rho\otimes \unit_{\caK_1}\otimes \unit_{\caK_2}-V_0V_0^*) U(\unit-V_0^*V_0)=0,\\
&\lmk z_\rho\otimes \unit_{\caK_1}\otimes \unit_{\caK_2}
\rmk
\cdot \lmk\unit-Z_{\pi(\caA)''}(V_0^*V_0)\rmk
=0.
\end{split}
\end{align}
We fix such $V_0$.
The last equality in (\ref{chairou}) implies
\begin{align}
 z_\rho\otimes \unit_{\caK_1}\otimes \unit_{\caK_2}
\le Z_{\pi(\caA)''}(V_0^*V_0),
\end{align}
while the first inequality in (\ref{chairou}) and the fact that $V_0$ and 
$z_\rho\otimes \unit_{\caK_1}\otimes\unit_{\caK_2}$ commute imply the opposite inequality.
Hence we obtain
\begin{align}
 z_\rho\otimes \unit_{\caK_1}\otimes \unit_{\caK_2}
=Z_{\pi(\caA)''}(V_0^*V_0),
\end{align}

Now let us consider the system $\caS$ of sets
of nonzero partial isometries $\{V_\lambda\}_{\lambda\in\ld}$
in $\pi(\caA\otimes\caB_1)''$
satisfying the following conditions:
\begin{description}
\item[(1)]
\begin{align}
\begin{split}
V_\lambda\pi(A_1\otimes B_1)=\rho(A_1\otimes B_1)V_\lambda,\quad \text{for all}\quad
A_1\in\caA_1,\quad
B_1\in \caB_1,
\end{split}
\end{align}
\item[(2)] projections
$\{V_\lambda V_\lambda^*\}_{\lambda\in\ld}$ are mutually orthogonal,
\item[(3)]
$V_\lambda V_\lambda^*\le z_\rho\otimes \unit_{\caK_1}\otimes \unit_{\caK_2}$
for any $\lambda\in\ld$,
\item[(4)]for any $\lambda\in\ld$,
$V_\lambda V_\lambda^*$ and $V_0V_0^*$
are orthogonal.
\end{description}
Note that $\caS$ is inductively ordered with respect to the inclusion order.
By Zorn's Lemma, there exists a maximal element
$S:=\{V_\lambda\}_{\lambda\in\ld}$ of $\caS$.
We fix such $S$.
Because $\caH\otimes\caK_1\otimes\caK_2$ is separable, from (2),
$S$ is at most countable, and we may write it
as $S:=\{V_n\}_{n=1}^N$ with $N\in\bbN$ of infinite.

We claim $\{V_n\}_{n=0}^N$ satisfies the required condition.
Note that $\{V_n\}_{n=0}^N$ satisfies conditions (\ref{daidai}) and (\ref{murasaki}).
It remains to show that it also satisfies (\ref{pink}).

By the property (1) to (4), and (\ref{chairou})
\begin{align}
\bar P:=V_0V_0^*+\sum_{n=1}^N V_n V_n^*
\end{align}
converges strongly and defines a projection in $\pi(\caA\otimes\caB_1)''\cap\rho(\caA_1\otimes \caB_1)'$,
satisfying $\bar P \le z_\rho\otimes \unit_{\caK_1}\otimes \unit_{\caK_2}$.
We would like to show that $P:= z_\rho\otimes \unit_{\caK_1}\otimes \unit_{\caK_2}-\bar P=0$
by contradiction.
Suppose that $P\neq 0$. Note from the definition that
$P$ belongs to $\pi(\caA\otimes\caB_1)''\cap\rho(\caA_1\otimes \caB_1)'$.
Applying Lemma \ref{orange} with $p$ replaced by $P$,
we obtain a partial isometry $V\in \pi(\caA\otimes\caB_1)''$
such that
\begin{align}
&VV^*\le P,\\
&V\pi(A_1\otimes B_1)=\rho(A_1\otimes B_1)V,\quad A_1\in\caA_1,\quad  B_1\in\caB_1,\\
&P\lmk\unit-Z_{\pi(\caA)''}(V^*V)\rmk=
P\lmk\unit-Z_{\pi(\caA)''}(V^*V)\rmk\lmk z_\rho\otimes\unit_{\caK_1}\otimes\unit_{\caK_2}\rmk=0.
\end{align}
 This $V$ is non-zero because if $V=0$,
 then the last equation implies $P=0$, which contradicts our assumption.
 Because of the first and second properties above, we see that
 $\{V\}\cup\{V_n\}_{n=1}^N$ is an element of $\caS$ which strictly majorizes
 $S=\{ V_n\}_{n=1}^N$. This contradicts the maximality of $S$.
 Hence we conclude that $P=0$.
\end{proof}
\begin{lem}\label{aiiro}
Consider the setting of Lemma \ref{orange} and assume that
$\caA_1$ satisfies \begin{align}\label{zin}
Z\lmk\pi_0(\caA)''\rmk=Z\lmk\pi_0(\caA_1)'\cap \pi_0(\caA)''\rmk.
\end{align}
Assume
further that
$\caB_1$ is of the form 
\begin{align}\label{hasibami}
\caB_1=\caB_1^{(1)}\otimes \caB_1^{(2)}
\end{align}
with $\caB_1^{(1)}$, $\caB_1^{(2)}$ infinite-dimensional UHF algebras,
$\caK_1$, $\pi_1$ of the form
\begin{align}\label{uguisu}
\begin{split}
\caK_1=\caK_1^{(1)}\otimes \caK_1^{(2)},\quad
\pi_1=\pi_1^{(1)}\otimes \pi_1^{(2)},
\end{split}
\end{align}
with $\pi_1^{(1)}$, $\pi_1^{(2)}$ irreducible representations of
$\caB_1^{(1)}$, $\caB_1^{(2)}$ on $\caK_1^{(1)}$, $\caK_1^{(2)}$
respectively.
Let $\{V_n\}_{n=0}^N$ (with $N\in\bbN$ or infinite)
as in Lemma \ref{yamabuki} and set
$p_n:=V_n V_n^*$, $q_n:=V_n^* V_n$.
Then, the following hold.
\begin{description}
\item[(i)]
There exists a sequence of mutually orthogonal projections
$\{r_n\}_{n=0}^N$ in $\pi(\caB_1^{(1)})''$ such that
$\sum_{n=0}^N r_n=\unit_{\caH\otimes\caK_1\otimes \caK_2}$
and each $r_n$ equivalent to $\unit$
 in $\pi(\caB_1^{(1)})''$.
\item[(ii)]
Operators $q_0 r_n$, $n=0,\ldots,N$ are mutually orthogonal projections
in $\pi\lmk\caA\otimes \caB_1^{(1)}\rmk''\cap \pi(\caA_1)'$.
\item[(iii)]
There exists a sequence $\{w_n\}_{n=1}^N$ of partial isometries in
$\pi\lmk\caA\otimes \caB_1^{(1)}\rmk''\cap \pi(\caA_1)'$ such that
\begin{align}\label{gin}
\begin{split}
w_nw_n^*=q_n,\quad \text{and}\quad w_n^* w_n\le q_0 r_0.
\end{split}
\end{align}
\item[(iv)]
There exists a sequence $\{W_n\}_{n=0}^N$
 of partial isometries in $\pi\lmk\caA\otimes \caB_1^{(1)}\rmk''\cap \pi(\caA_1)'$
 such that 
 \begin{align}\label{gold}
 \begin{split}
 &W_n^* W_n=r_n\lmk z_\rho\otimes\unit_{\caK_1}\otimes \unit_{\caK_2}\rmk,\\
 &W_nW_n^*
 =\left\{
 \begin{gathered}
 w_n^* w_n+q_0r_n,\quad n\ge 1,\\
 q_0r_0,\quad n=0
 \end{gathered}
 \right..
 \end{split}
 \end{align}
 \item[(v)]
 The sum
 \begin{align}\label{dou}
 W:=\sum_{n=0}^N V_0 q_0 r_n W_n
 +\sum_{n=1}^N V_n w_n W_n\in \pi\lmk\caA\otimes\caB_1\rmk''
 \end{align}
 converges in the strong$*$-topology and satisfies
 \begin{align}
& W^*W=WW^*=z_\rho\otimes \unit_{\caK_1}\otimes \unit_{\caK_2},\label{asagi}\\
 &W\pi(A_1\otimes B) =\rho(A_1\otimes B) W,\quad \text{for all}\quad A_1\in \caA_1,\quad B\in \caB_1^{(2)}\otimes\caB_2
\label{momoiro}
 \end{align}
 
\end{description}
\end{lem}
\begin{proof}
Set
\begin{align}
\caM:=\lmk \pi_0(\caA)''\cap\pi_0(\caA_1)'\rmk \bar\otimes 
\pi_1^{(1)}(\caB_1^{(1)})''\bar\otimes \bbC\unit_{\caK_1^{(2)}}\bar\otimes \bbC\unit_{\caK_2}
=\pi\lmk \caA\otimes \caB_1^{(1)}\rmk''\cap\pi(\caA_1)'.
\end{align}
Note that 
\begin{align}\label{india}
 \begin{split}
  &Z\lmk\caM \rmk= Z\lmk \pi_0(\caA)''\cap\pi_0(\caA_1)'\rmk \bar\otimes \bbC\unit_{\caK_1^{(1)}}
 \bar\otimes\bbC\unit_{\caK_1^{(2)}\otimes \caK_2}
 =Z\lmk \pi_0(\caA)''\rmk \bar\otimes \bbC\unit_{\caK_1}\bar\otimes \bbC\unit_{\caK_2}
 =Z\lmk\pi(\caA)''\rmk,
   \end{split}
 \end{align}
by the assumption (\ref{zin}) and the irreducibility of $\pi_1^{(1)}$.
Note also  
\begin{align}
\begin{split}
\pi(\caB_1^{(1)})''=\bbC\unit_{\caH_0}\bar\otimes \caB(\caK_1^{(1)})\bar\otimes\bbC \unit_{\caK_1^{(2)}}\bar\otimes \bbC\unit_{\caK_2}\subset \caM.
\end{split}
\end{align}
From this, for any pair of projections $r,\tilde r\in \pi(\caB_1^{(1)})''$ that are equivalent in $ \pi(\caB_1^{(1)})''$ and any
projection 
$$q\in \lmk \pi_0(\caA)''\cap\pi_0(\caA_1)'\rmk \bar\otimes \bbC\unit_{\caK_1}\bar\otimes\unit_{\caK_2},$$
$qr$ and $q\tilde r$ are projections in $\caM$ that are equivalent in $\caM$.
In particular, from (\ref{india}),
 if $r\in \pi(\caB_1^{(1)})''$ is an infinite projection in $\pi(\caB_1^{(1)})''$
and $q\in \lmk \pi_0(\caA)''\cap\pi_0(\caA_1)'\rmk \bar\otimes \bbC\unit_{\caK_1}\bar\otimes\unit_{\caK_2}$
is a nonzero projection, then $qr$ is properly infinite in  $\caM$.

We have $p_n\in \rho(\caA_1\otimes \caB_1)'\cap \pi(\caA\otimes\caB_1)''$ and 
$$q_n\in \pi(\caA\otimes\caB_1)''\cap \pi(\caA_1\otimes \caB_1)'
=\lmk \pi_0(\caA)''\cap\pi_0(\caA_1)'\rmk\bar\otimes \bbC\unit_{\caK_1}\bar\otimes \bbC\unit_{\caK_2}\subset \caM$$
 by (\ref{daidai}) and the irreducibility of $\pi_1$.

(i) is trivial because $\pi(\caB_1^{(1)})''$ is a type $I_\infty$ factor due to the irreducibility of
$\pi_1^{(1)}$.
Note that $r_n$ is infinite in 
$\pi(\caB_1^{(1)})''$.
(ii) is trivial based on the above observation.\\
(iii)
Because $r_n$ is infinite in $\pi(\caB_1^{(1)})''$ and $q_0\in \lmk \pi_0(\caA)''\cap\pi_0(\caA_1)'\rmk\bar\otimes \bbC\unit_{\caK_1}\bar\otimes \bbC\unit_{\caK_2}$ is a nonzero projection,
from the above observation, 
$q_0 r_n$ is properly infinite in $\caM$, $n=0,\ldots,N$.
Because $r_n$ and $\unit$ are equivalent in $\pi(\caB_1^{(1)})''$,
$q_0r_n$ and $q_0$ are equivalent in $\caM$.
Therefore, by Proposition 6.2.8 of \cite{KR}, we get
\begin{align}
\begin{split}
Z_{\caM}\lmk q_0 r_n\rmk
=Z_{\caM}\lmk q_0 \rmk=Z_{\pi(\caA)''}(q_0)
=z_\rho\otimes\unit_{\caK_1}\otimes \unit_{\caK_2},
\end{split}
\end{align}
for any $n=0,\ldots, N$ by (\ref{india}) and (\ref{murasaki}).
Because of this and $q_n=V_n^* V_n$ and $p_n=V_nV_n^*\le z_\rho\otimes\unit_{\caK_1}\otimes \unit_{\caK_2}$,
we have $q_n\le z_\rho\otimes\unit_{\caK_1}\otimes \unit_{\caK_2}$
and
\begin{align}\label{oudo}
Z_{\caM}\lmk q_n\rmk
\le z_\rho\otimes\unit_{\caK_1}\otimes \unit_{\caK_2}
=Z_{\caM}\lmk q_0 r_0\rmk.
\end{align}
Hence, $q_0 r_0$ is a properly infinite projection
in $\caM$ and $q_n$ is a projection in $\caM$
with the central carriers satisfying (\ref{oudo}).
Therefore, by Theorem 6.3.4 of \cite{KR},
there exists $w_n\in \caM$ 
satisfying (\ref{gin}).

(iv) Because $r_n$ is infinite in $\pi(\caB_1^{(1)})''$,
 $\lmk z_\rho\otimes \unit_{\caK_1}\otimes \unit_{\caK_2}\rmk r_n$
is properly infinite in $\caM$.
We also have
$Z_{\caM}\lmk \lmk z_\rho\otimes \unit_{\caK_1}\otimes \unit_{\caK_2}\rmk r_n\rmk= z_\rho\otimes \unit_{\caK_1}\otimes \unit_{\caK_2}$ because of the equivalence
of $r_n$ and $\unit$ in $\pi(\caB_1^{(1)})''$.

Set
\begin{align}
\begin{split}
\hat q_n:=\left\{
 \begin{gathered}
 w_n^* w_n+q_0r_n,\quad n\ge 1,\\
 q_0r_0,\quad n=0
 \end{gathered}
 \right..
\end{split}
\end{align}
They are projections in $\caM$
because
$w_n^*w_n\le q_0r_0$ and $q_0r_0$ and $q_0r_n$
are mutually orthogonal for $n\ge 1$.
The central carriers of $\hat q_n$ in $\caM$
are $z_\rho\otimes \unit_{\caK_1}\otimes \unit_{\caK_2}$ for all $n$
because 
\begin{align}
\begin{split}
z_\rho\otimes \unit_{\caK_1}\otimes \unit_{\caK_2}
=
Z_{\caM}\lmk q_0r_n\rmk \le Z_{\caM}\lmk \hat q_n\rmk\le 
Z_{\caM }\lmk q_0\rmk
=z_\rho\otimes \unit_{\caK_1}\otimes \unit_{\caK_2}.
\end{split}
\end{align}
We know that $\hat q_0=q_0 r_0$ is properly infinite in $\caM$
from above.
Now we show that $\hat q_n$ is properly infinite in $\caM$ for $n\ge 1$.
For any projection $z\in Z\lmk\caM \rmk$
, we have $zq_0r_n\le z\hat q_n$.
If $z\hat q_n$ is finite in $\caM$, $zq_0r_n$ has to be finite in $\caM$ by Proposition 6.3.2 of \cite{KR}.
Because $q_0 r_n$ is properly infinite in $\caM$,
this means $zq_0r_n=0$, and we obtain
\begin{align}
z_\rho\otimes\unit_{\caK_1}\otimes \unit_{\caK_2}=
Z_{\caM}\lmk q_0r_n \rmk\le \unit-z.
\end{align}
This implies
\begin{align}
\hat q_n\le q_0\le Z_{\caM}\lmk q_0\rmk
= z_\rho\otimes\unit_{\caK_1}\otimes \unit_{\caK_2}
\le \unit-z,
\end{align}
which implies $\hat q_n z=0$. Hence $\hat q_n$ is properly infinite in $\caM$.

As a result, both $\hat q_n$ and $\lmk z_\rho\otimes \unit_{\caK_1}\otimes \unit_{\caK_2}\rmk r_n$
are properly infinite projections in $\caM$
and with the same central carrier $z_\rho\otimes \unit_{\caK_1}\otimes \unit_{\caK_2}$ in $\caM$.
Therefore, by Corollary 6.3.5 of \cite{KR}, they are equivalent in $\caM$ hence there exists $W_n\in \caM$
satisfying (\ref{gold}).

(v)
Note that 
\begin{align}
\begin{split}
&V_0 q_0 r_n W_n \lmk  V_0 q_0 r_n W_n\rmk^*
=V_0q_0 r_n V_0^*,\quad n\ge 0,\\
& V_n w_n W_n\lmk V_n w_n W_n\rmk^*
=V_n q_n V_n^*,\quad n\ge 1,
\end{split}
\end{align}
are mutually orthogonal projections in
$\pi\lmk\caA\otimes \caB_1\rmk''$.
They further sum up to
\begin{align}\label{chairo}
\begin{split}
\sum_{n=0}^N V_0q_0 r_n V_0^*+
\sum_{n=1}^NV_n q_n V_n^*
=V_0 q_0 V_0^* +\sum_{n=1}^N V_n q_n V_n^*
=\sum_{n=0}^N V_nV_n^*
=z_\rho\otimes\unit_{\caK_1}\otimes \unit_{\caK_2}
\end{split}
\end{align}
On the other hand, 
\begin{align}
\begin{split}
&\lmk  V_0 q_0 r_n W_n\rmk^* V_0 q_0 r_n W_n 
=W_n^* r_n q_0 r_n W_n,\quad n\ge 0\\
&\lmk V_n w_n W_n\rmk^* V_n w_n W_n
=W_n^* w_n^* w_n W_n,\quad n\ge 1
\end{split}
\end{align}
are mutually orthogonal projections in
$\pi\lmk\caA\otimes \caB_1\rmk''$,
which sum up to
\begin{align}
\begin{split}
\sum_{n=0}^N W_n^* r_n q_0 r_n W_n
+\sum_{n=1}^N W_n^* w_n^* w_n W_n
=\sum_{n=0}^N W_n^* W_n
=\sum_{n=0}^N r_n\lmk z_\rho\otimes\unit_{\caK_1}\otimes \unit_{\caK_2}\rmk
=z_\rho\otimes\unit_{\caK_1}\otimes \unit_{\caK_2}.
\end{split}
\end{align}
Hence (\ref{dou}) converges in the strong $*$-topology, and
we obtain a partial isometry $W\in \pi(\caA\otimes\caB_1)''$
such that $W^*W=WW^*=z_\rho\otimes\unit_{\caK_1}\otimes \unit_{\caK_2}$.
Because 
all of $q_0$, $r_n$, $W_n$, $w_n$ belong to
$\caM=\pi\lmk \caA\otimes \caB_1^{(1)}\rmk''\cap\pi(\caA_1)'$, they 
commute with $\pi(\caA_1\otimes\caB_1^{(2)}\otimes\caB_2)$.
The partial isometries $V_n\in \pi\lmk \caA\otimes \caB_1\rmk''$
satisfy (\ref{daidai}) and commute with $\pi(\caB_2)''$.
This proves (\ref{momoiro}).
Hence we have completed the proof of (v).
\end{proof}
\subsection{Stabilization}
Now, we come back to our setting and prove Theorem \ref{stab}, using the Lemmas in the previous subsection.
First, we note the following basic fact.
\begin{lem}\label{wani}
Let $\varphi$ be a state on a $2$-dimensional quantum spin system $\caA$  satisfying the approximate Haag duality.
Let $\caB$ be a two-dimensional quantum spin system and $\psi$
a pure infinite tensor product state on $\caB$.
Then the state $\varphi\otimes \psi$ on
$\caA\otimes \caB$
has properly infinite cone algebras and satisfies the approximate Haag duality.
\end{lem}
\begin{proof}
Let $(\caH_\varphi, \pi_\varphi)$, $(\caH_\psi,\pi_\psi)$ be GNS representations of
$\varphi$ and $\psi$.
Their tensor product $(\caH,\pi):=(\caH_\varphi\otimes\caH_\psi,\pi_\varphi\otimes \pi_\psi)$
is a GNS representation of $\varphi\otimes \psi$.

For any cone $\ld$ in $\bbZ^2$, we have 
\begin{align}
\begin{split}
\unit_{\varphi} \bar\otimes \pi_\psi \lmk \caB_{\ld}\rmk''
\subset \pi_\varphi(\caA_\ld)\bar\otimes \pi_\psi\lmk\caB_\ld\rmk''.
\end{split}
\end{align}
Because $\unit_{\caH_\varphi} \otimes \pi_\psi \lmk \caB_{\ld}\rmk''$ is properly infinite,
$\pi_\varphi(\caA_\ld)\bar\otimes \pi_\psi\lmk\caB_\ld\rmk''$ is also properly infinite.

Next we show that 
$\varphi\otimes \psi$ satisfies the approximate Haag duality.
For any $\zeta\in (0,\pi)$ and $0<\varepsilon<\frac 14(\pi-\zeta)$,
let $R_{\zeta,\varepsilon}$, $f_{\zeta,\varepsilon}$ be as in Definition \ref{AHdef}
for the state $\varphi$.
For any cone $\Lambda$ with $\lv\Lambda\rv=2\zeta$,
there exists a unitary $U_{\Lambda,\varepsilon}\in \caU\lmk\pi_\varphi\lmk\caA\rmk''\rmk$
such that
\[
\pi_\varphi(\caA_{\Lambda^c})'\cap \pi_\varphi\lmk \caA_{\bbZ^2}\rmk''
\subset \Ad \lmk U_{\Lambda,\varepsilon}\rmk\lmk
\pi_\varphi\lmk\caA_{\Lambda_{\varepsilon}(-R_{\zeta,\varepsilon})}\rmk''
\rmk,\]
 and for any $t\ge 0$, there exists a unitary $U_{\Lambda,\varepsilon, t}\in \caU\lmk
\pi_\varphi\lmk \caA_{\Lambda_{2\varepsilon}(-t)}\rmk''\rmk$
such that
\begin{align}
\lV U_{\Lambda,\varepsilon, t}-U_{\Lambda,\varepsilon}\rV
\le f_{\zeta,\varepsilon}(t).
\end{align}
Then we have
\begin{align}
\begin{split}
&U_{\Lambda,\varepsilon}\otimes\unit
\in \caU\lmk
\pi(\caA\otimes{\caB})''
\rmk,\\
&U_{\Lambda,\varepsilon, t}\otimes\unit
\in \caU\lmk
\pi\lmk \lmk \caA\otimes{\caB}\rmk_{\Lambda_{2\varepsilon}(-t)}\rmk''
\rmk,\\
&\lV U_{\Lambda,\varepsilon}\otimes\unit- U_{\Lambda,\varepsilon, t}\otimes\unit\rV
\le f_{\zeta,\varepsilon}(t)
\end{split}
\end{align}
and
\begin{align}
\begin{split}
&\pi\lmk\lmk \caA\otimes{\caB}\rmk_{\Lambda^c}\rmk'\cap
\pi\lmk \caA\otimes{\caB}\rmk''
=\lmk \pi_\varphi\lmk\caA_{\Lambda^c}\rmk'\cap \pi_\varphi({ \caA})''\rmk
\bar\otimes \lmk \pi_\psi\lmk{\caB}_{\Lambda^c}\rmk'\cap\pi_\psi\lmk {\caB}\rmk''\rmk\\
&\subset 
\Ad\lmk U_{\Lambda,\varepsilon}\otimes\unit\rmk
\lmk 
\pi_\varphi\lmk\caA_{\Lambda_{\varepsilon}(-R_{\zeta,\varepsilon})}\rmk''
\bar\otimes \pi_\psi\lmk {\caB}_\Lambda\rmk''
\rmk\\
&\subset
\Ad\lmk U_{\Lambda,\varepsilon}\otimes\unit\rmk
\lmk
\pi\lmk\lmk \caA\otimes{\caB}\rmk_{\Lambda_{\varepsilon}(-R_{\zeta,\varepsilon})}\rmk''
\rmk.
\end{split}
\end{align}
Hence $\omega$ satisfies the approximate Haag duality.
\end{proof}

Now, we would like to show that
$\fol{\omega\otimes\psi_1\otimes \psi_2}{\omega\otimes\psi_1}$
is an equivalence if $\psi_1$, $\psi_2$ are pure infinite tensor product states
on two-dimensional quantum spin systems.
In order to do that we prepare the following.

\begin{lem}\label{lem46}
Let $\ld_0$ be a cone.
Let $\caB_1$, $\caB_2$ be two-dimensional quantum spin systems and $\psi_1$, $\psi_2$
pure infinite tensor product states on $\caB_1$, $\caB_2$, respectively.
Let $\omega$ be a state on a two-dimensional quantum spin system $\caA$.
If $\rho\in \ools{\omega\otimes\psi_1\otimes\psi_2}$
with $\rho(\unit)\in Z\lmk
\lmk\pi_{\omega}\otimes\pi_{\psi_1}\otimes\pi_{\psi_2}\rmk\lmk\caA\otimes\caB_1\otimes\caB_2\rmk''\rmk$,
the following hold.
\begin{description}
\item[(i)]
There exists a partial isometry
$V$ in $\pi_{\omega}(\caA)''\bar \otimes\pi_{\psi_1}(\caB_1)''\bar
\otimes\pi_{\psi_2}(\caB_2)''$
such that 
\begin{align}
\begin{split}
&V^*V=VV^*=\rho(\unit),\\
&V\pi(B_2)=\rho(B_2) V,\quad B_2\in \caB_2,
\end{split}
\end{align}
and a representation $\rho_1$ of
$\caA\otimes\caB_1$ on $\caH_\omega\otimes \caH_{\psi_1}$
such that
\begin{align}
\rho_1(A\otimes B_1)\otimes \unit_{\caH_{\psi_2}}
=\Ad V^*\rho \lmk A\otimes B_1 \rmk,\quad
A\in \caA,\quad B_1 \in \caB_1.
\end{align}
\item[(ii)]
For any cone $\ld$, there exists a partial isometry
$V_\ld\in \pi_{\omega}\lmk\caA\rmk''\bar\otimes\pi_{\psi_1}(\caB_1)''$
such that
\begin{align}
&V_\ld^* V_\ld =V_{\ld} V_{\ld}^*=\rho_1(\unit),\\
&V_\ld \lmk\pi_\omega\otimes \pi_{\psi_{1}}\rmk(A\otimes B_1)
=\rho_1\lmk A\otimes B_1\rmk V_\ld,\quad
A\in\caA_{\ld^c},\quad B_1\in (\caB_1)_{\ld^c}.
\end{align}
\item[(iii)]
The formula
$\gamma:=\Ad V_{\ld_0}^*\circ\rho_1$ (with $V_{\ld_0}$ given in (ii)) defines
an element $\gamma$ of $\ools{\omega\otimes\psi_1}$.
\end{description}
\end{lem}
\begin{proof}
We set $\pi:=\pi_\omega\otimes\pi_{\psi_1}\otimes\pi_{\psi_2}$.
Let $z_\rho\in Z\lmk\pi_\omega(\caA)''\rmk$ be a projection given by 
$$z_\rho\otimes\unit_{\caH_{\psi_1}}\otimes\unit_{\caH_{\psi_2}}=\rho(\unit)\in Z\lmk\pi(\caA\otimes\caB_1\otimes\caB_2)''\rmk
=Z\lmk\pi_\omega(\caA)''\rmk\bar\otimes\bbC\unit_{\caH_{\psi_1}}\bar\otimes\bbC\unit_{\caH_{\psi_2}}.$$

(i)
We apply Lemma \ref{aiiro} 
with $\caA$, $\caA_1$, $\caB_1^{(1)}$, $\caB_1^{(2)}$, $\caB_2$,
$\caH_0$, $\caK_1^{(1)}$, $\caK_1^{(2)}$,
$\caK_2$, $\pi_0$, $\pi_1^{(1)}$, $\pi_1^{(2)}$, $\pi_2$, $U$, $\rho$, $z_\rho$,
replaced by
$\caA$, $\bbC\unit_{\caA}$, $(\caB_1)_{\ld_0}$, $(\caB_2)_{\ld_0}$, $\lmk\caB_1\otimes \caB_2\rmk_{\ld_0^c}$,
$\caH_\omega$, $\caH_{\psi_1\vert_{(\caB_1)_{\ld_0}}}$, 
$\caH_{\psi_2\vert_{(\caB_2)_{\ld_0}}}$, 
$\caH_{\lmk\psi_1\otimes\psi_2\rmk\vert_{(\caB_1\otimes\caB_2)_{\ld_0^c}}}$,
$\pi_\omega$, 
$\pi_{\psi_1\vert_{(\caB_1)_{\ld_0}}}$, 
$\pi_{\psi_2\vert_{(\caB_2)_{\ld_0}}}$, 
$\pi_{\lmk\psi_1\otimes\psi_2\rmk\vert_{(\caB_1\otimes\caB_2)_{\ld_0^c}}}$,
$V_{\rho\ld_0^c}\in\caV_{\rho\ld_0^c}$, $\rho$, $z_\rho$.
%Note that
%\begin{align}
%\begin{split}
%&Z\lmk\lmk \pi_\omega\otimes\pi_{\psi\vert_{\caB_\Gamma}}\rmk\lmk \caA\otimes\caB_\Gamma\rmk''\rmk
%=Z\lmk\pi_\omega(\caA)''\rmk\bar\otimes\unit_{\psi\vert_{\caB_\Gamma}}
%=Z\lmk
%\pi_\omega\lmk\caA_{\Gamma^c}\rmk'\cap\pi_\omega(\caA)''
%\rmk
%\bar \otimes
%\unit_{\psi\vert_{\caB_\Gamma}}\\
%&=Z\lmk
%\lmk
% \pi_\omega\otimes\pi_{\psi\vert_{\caB_\Gamma}}
%\lmk \caA_{\Gamma^c}\rmk\rmk'
%\cap
%\lmk \pi_\omega\otimes\pi_{\psi\vert_{\caB_\Gamma}}\rmk\lmk
%\caA\otimes\caB_\Gamma
%\rmk''
%\rmk
%\end{split}
%\end{align}
%by (\ref{zzz}). This corresponds to the condition (\ref{zin}) of Lemma \ref{aiiro}.
By the definition of $\caV_{\rho\ld_0^c}$ and $\ools{\omega\otimes\psi_1\otimes\psi_2}$,
the conditions of Lemma \ref{aiiro} are satisfied.

Applying  Lemma \ref{aiiro},
we obtain $V\in   \pi \lmk \caA\otimes (\caB_1)_{\ld_0}\otimes (\caB_2)_{\ld_0}\rmk''$
such that 
\begin{align}
\begin{split}
&V^*V=VV^*=\rho(\unit),\\
&V\pi(B)=\rho(B)V,\quad B\in (\caB_2)_{\ld_0}\otimes (\caB_1\otimes\caB_2)_{\ld_0^c}
=(\caB_1)_{\ld_0^c}\otimes \caB_2.
\end{split}
\end{align}
From the second equality, 
for any $A\in\caA$, $B_1\in \caB_1$ and $B_2\in\caB_2$, we see
\begin{align}
\begin{split}
&\left[
\Ad V^*\rho(A\otimes B_1), \pi(B_2)
\right]
=V^* \rho(A\otimes B_1) V \pi(B_2)-
\pi(B_2) V^* \rho(A\otimes B_1) V \\
&=V^* \rho(A\otimes B_1) \rho(B_2) V
-V^* \rho(B_2)\rho(A\otimes B_1) V 
=V^* \left[\rho(A\otimes B_1), \rho(B_2) \right]V=0.
\end{split}
\end{align}
Hence we obtain
\begin{align}
\begin{split}
\Ad V^*\rho\lmk \caA\otimes \caB_1\rmk
\subset \pi\lmk\caA\otimes\caB_1\otimes\caB_2\rmk''
\cap \pi(\caB_2)'
=\pi_\omega(\caA)''\bar \otimes\pi_{\psi_1}(\caB_1)''\bar \otimes \bbC\unit_{\caH_{\psi_2}} 
\end{split}
\end{align}
and we obtain a linear operator
$\rho_1: \caA\otimes \caB_1\to\caB(\caH_\omega\otimes\caH_{\psi_1})$
satisfying
\begin{align}
\rho_1(A\otimes B_1)\otimes \unit_{\caH_{\psi_2}}
=\Ad V^*\rho \lmk A\otimes B_1 \rmk,\quad
A\in \caA,\quad B_1 \in \caB_1.
\end{align}
Because $VV^*=\rho(\unit)$, this $\rho_1$ is a representation of
$\caA\otimes \caB_1$ on $\caH_\omega\otimes\caH_{\psi_1}$.

(ii)
For any cone $\ld$, we fix cones $\Gamma$, $D$
satisfying $\Gamma\cap D=\emptyset$, $\Gamma,D\subset\ld$ and 
an operator
$V_{\rho\Gamma}\in \caV_{\rho\Gamma}\subset \pi(\caA\otimes\caB_1\otimes\caB_2)''$.
We set $U:=V^* V_{\rho\Gamma}\in  \pi(\caA\otimes\caB_1\otimes\caB_2)''$, with $V$ in (i).
Then we have
\begin{align}\label{shima}
\begin{split}
&\Ad U\pi\lmk A\otimes B_1 \rmk
=\Ad V^* V_{\rho\Gamma} \pi(A\otimes B_1)
=\Ad V^* \rho(A\otimes B_1)
=\rho_1\lmk A\otimes B_1\rmk \otimes\unit_{\caH_{\psi_2}},\\
&\quad\text{for any}A\in\caA_{\Gamma^c} \;\text{and}\;B_1\in(\caB_1)_{\Gamma^c}.
\end{split}
\end{align}
We also have
\begin{align}\label{koushi}
\begin{split}
&U^* U= V_{\rho\Gamma}^* V V^* V_{\rho\Gamma}
=V_{\rho\Gamma}^* \rho(\unit) V_{\rho\Gamma}
=V_{\rho\Gamma}^* V_{\rho\Gamma}
\in\pi\lmk\lmk \caA\otimes\caB_1\otimes\caB_2\rmk_{\Gamma^c}\rmk'
\subset \pi\lmk \lmk \caA\otimes\caB_1\rmk_{\Gamma^c}\rmk',\\
&\rho_1(\unit_{\caA\otimes\caB_1})\otimes\unit_{\caH_{\psi_2}}=UU^*=V^* V_{\rho\Gamma}V_{\rho\Gamma}^* V
=V^*V=\rho(\unit)
=z_\rho\otimes\unit_{\caH_{\psi_1}}\otimes \unit_{\caH_{\psi_2}}.
\end{split}
\end{align}
Now define a representation $\tilde\rho$
 of $\caA\otimes\caB_1\otimes\caB_2$
 on $\caH_\omega\otimes \caH_{\psi_1}\otimes\caH_{\psi_2}$ by
\begin{align}
\tilde\rho:=\rho_1\otimes \pi_{\psi_2}.
\end{align}

Now we apply Lemma \ref{aiiro}
with $\caA$, $\caA_1$
$\caB_1^{(1)}$, $\caB_1^{(2)}$,
$\caB_2$, $\caH$, $\pi_0$,
$\caK_1^{(1)}$, $\pi_1^{(1)}$,
$\caK_1^{(2)}$, $\pi_1^{(2)}$,
$\caK_2$, $\pi_2$,
$U$, $\rho$, $z_\rho$
replaced by
$\caA\otimes(\caB_1)_{\Gamma}$, 
$\caA_{\Gamma^c}$,
$(\caB_1)_{\Gamma^c\cap \ld}$,
$(\caB_1)_{\ld^c}$,
$\caB_2$,
$\caH_\omega\otimes \caH_{\psi_1\vert_{(\caB_1)_\Gamma}}$,
$\pi_\omega\otimes\pi_{\psi_1\vert_{(\caB_1)_\Gamma}}$,
$\caH_{\psi_1\vert_{(\caB_1)_{\ld \cap\Gamma^c}}}$,
$\pi_{\psi_1\vert_{(\caB_1)_{\ld\cap\Gamma^c}}}$,
$\caH_{\psi_1\vert_{(\caB_1)_{\ld^c}}}$,
$\pi_{\psi_1\vert_{(\caB_1)_{\ld^c}}}$,
$\caH_{\psi_2}$, $\pi_{\psi_2}$,
$U$, $\tilde\rho$, 
$z_\rho\otimes\unit_{\caH_{\psi_1\vert_{(\caB_1)_{\Gamma}}}}$.
Note from (\ref{shima}) (\ref{koushi}) and the definition of $\tilde\rho$ and
\begin{align}
\begin{split}
&Z\lmk\lmk \pi_\omega\otimes\pi_{\psi_1\vert_{(\caB_1)_\Gamma}}\rmk\lmk \caA\otimes(\caB_1)_\Gamma\rmk''\rmk
=Z\lmk\pi_\omega(\caA)''\rmk\bar\otimes\bbC\unit_{\psi_1\vert_{(\caB_1)_\Gamma}}
=Z\lmk
\pi_\omega\lmk\caA_{\Gamma^c}\rmk'\cap\pi_\omega(\caA)''
\rmk
\bar \otimes
\bbC \unit_{\psi_1\vert_{(\caB_1)_\Gamma}}\\
&=Z\lmk
\lmk
 \pi_\omega\otimes\pi_{\psi_1\vert_{(\caB_1)_\Gamma}}
\lmk \caA_{\Gamma^c}\rmk\rmk'
\cap
\lmk \pi_\omega\otimes\pi_{\psi_1\vert_{(\caB_1)_\Gamma}}\rmk\lmk
\caA\otimes(\caB_1)_\Gamma
\rmk''
\rmk
\end{split}
\end{align}
by (\ref{zzz}),
the conditions required in Lemma \ref{aiiro} hold.
Hence we may apply Lemma \ref{aiiro}.

Applying Lemma \ref{aiiro}, we obtain 
\begin{align}\label{men}
W\in \pi(\caA\otimes(\caB_1)_{\Gamma}\otimes(\caB_1)_{\Gamma^c})''
=\pi_\omega(\caA)''\bar\otimes\pi_{\psi_1}(\caB_1)''\bar\otimes\unit_{\caH_{\psi_2}}
\end{align}
such that
\begin{align}
\begin{split}
&W^* W=WW^*=\tilde \rho(\unit)
=z_\rho\otimes\unit_{\psi_1}\otimes\unit_{\psi_2},\\
&W\pi(A\otimes B_1)=\tilde\rho(A\otimes B_1) W
=\lmk \rho_1(A\otimes B_1)\otimes\unit_{\caH_{\psi_2}}\rmk W,\quad
\text{for all}\; A\in\caA_{\Gamma^c},\; B_1\in (\caB_1)_{\ld^c}.
\end{split}
\end{align}
From (\ref{men}), $W$ is of the form
$W=V_\ld\otimes \unit_{\caH_{\psi_2}}$
with $V_\ld\in \pi_\omega(\caA)''\bar\otimes\pi_{\psi_1}(\caB_1)''$.
From above, this $V_\ld$ satisfies
\begin{align}
&V_\ld^*V_\ld=V_\ld V_\ld^*=z_\rho\otimes \unit_{\caH_{\psi_1}}=\rho_1(\unit),\\
&V_\ld\lmk \pi_\omega(A)\otimes \pi_{\psi_1}(B_1)\rmk
=\rho_1(A\otimes B_1)V_\ld,\quad
A\in \caA_{\ld^c}\subset \caA_{\Gamma^c}, \; B_1\in (\caB_1)_{\ld^c}.\label{himo}
\end{align}

(iii)
The formula
$\gamma:=\Ad V_{\ld_0}^*\circ\rho_1$
defines a represntation of $\caA\otimes\caB_1$ on $\caH_{\omega}\otimes\caH_{\psi_1}$
because $V_{\ld_0}V_{\ld_0}^*=\rho_1(\unit)$.
For any cone $\ld$, we claim that
$V_{\ld_0}^* V_{\ld}\in\caV_{\gamma\ld}$.
In fact, $V_{\ld_0}^* V_{\ld}\in \pi_\omega(\caA)''\bar\otimes\pi_{\psi_1}(\caB_1)''$
from above and
\begin{align}
\begin{split}
&\Ad\lmk V_{\ld_0}^* V_{\ld}\rmk \circ
\lmk \pi_\omega\otimes\pi_{\psi_1}\rmk(X)
= V_{\ld_0}^* V_{\ld} \lmk \pi_\omega\otimes\pi_{\psi_1}\rmk(X) V_{\ld}^* V_{\ld_0}\\
&= V_{\ld_0}^* \rho_1(X) V_{\ld}V_{\ld}^*  V_{\ld_0}
= V_{\ld_0}^* \rho_1(X) \rho_1(\unit)  V_{\ld_0}
=\gamma(X),\\
&\quad X\in \lmk\caA\otimes\caB_1\rmk_{\ld^c}
\end{split}
\end{align}
and
\begin{align}
\begin{split}
&\lmk  V_{\ld_0}^* V_{\ld}\rmk^*  V_{\ld_0}^* V_{\ld}=
V_\ld^* V_{\ld_0}V_{\ld_0}^* V_\ld 
=V_\ld^* \rho_1(\unit) V_\ld
= V_\ld^*V_\ld=\rho_1(\unit)
=z_\rho\otimes\unit_{\caH_{\psi_1}}\\
&\in Z\lmk\pi_\omega(\caA)''\rmk\bar \otimes\unit_{\caH_{\psi_1}}
\subset \lmk \pi_\omega\otimes\pi_{\psi_1}\lmk \lmk \caA\otimes \caB\rmk_{\ld^c}\rmk\rmk'.
\end{split}
\end{align}
Furthermore, $\gamma$ satisfies
\begin{align}
\begin{split}
\gamma(X)=V_{\ld_0}^* \rho_1(X) V_{\ld_0}
=V_{\ld_0}^* V_{\ld_0} \pi_{\omega}\otimes\pi_{\psi_1}(X)
=\lmk \pi_\omega\otimes\pi_{\psi_1} (X) \rmk \rho_1(\unit)
=\lmk \pi_\omega\otimes\pi_{\psi_1} (X)\rmk\gamma(\unit),
\end{split}
\end{align}
for any $X\in \lmk \caA\otimes\caB_1\rmk_{\ld_{0}^c}$, because of (\ref{himo}).
Hence we obtain $\gamma\in \ools{\omega\otimes\psi_1}$.
\end{proof}
\begin{proofof}[Theorem \ref{stab}]
Now, we are ready to prove Theorem \ref{stab}.
First of all, from Lemma \ref{wani}, both
$\omega\otimes \psi_1$ and $\omega\otimes\psi_1\otimes\psi_2$
have properly infinite cone algebras and satisfy the approximate Haag duality.
Hence we obtain a strict braided tensor functor
$\fol{\omega\otimes\psi_1\otimes\psi_2}{\omega\otimes\psi_1}:
\cols{\omega\otimes\psi_1}\to \cols{\omega\otimes\psi_1\otimes\psi_2}$
by Theorem \ref{cofe}.
Because $\lmk \omega\otimes\psi_1\otimes\psi_2\rmk\vert_{\caB_2}=\psi_2$
is pure, from Lemma \ref{washi}, $\fol{\omega\otimes\psi_1\otimes\psi_2}{\omega\otimes\psi_1}$
is fully faithful.
What remains to be shown is that $\fol{\omega\otimes\psi_1\otimes\psi_2}{\omega\otimes\psi_1}$
is essentially surjective.

Set $\pi_{\omega\otimes\psi_1\otimes\psi_2}:=\pi_\omega\otimes\pi_{\psi_1}\otimes \pi_{\psi_{2}}$.
To show the essential surjectivity, let $\rho\in \Obj \cols{\omega\otimes\psi_1\otimes\psi_2}$.
Because of Lemma \ref{shachi},
 it suffices to consider the case that $\rho(\unit)$ belongs to the center
$$Z\lmk \pi_{\omega\otimes\psi_1\otimes\psi_2}\lmk
\caA\otimes \caB_1\otimes\caB_2\rmk''\rmk.$$
By Lemma \ref{lem46},
there exist a represention $\rho_1$ of $\caA\otimes\caB_1$
on $\caH_\omega\otimes\caH_{\psi_1}$,
 partial isometries $V\in \pi_{\omega\otimes\psi_1\otimes\psi_2}\lmk
\caA\otimes \caB_1\otimes\caB_2\rmk''$ and
$V_{\ld_0}\in \lmk \pi_\omega\otimes\pi_{\psi_1}\rmk(\caA\otimes\caB_1)''$
such that
\begin{align}
\begin{split}
&\gamma:=\Ad V_{\ld_0}^*\rho_1\in \ools{\omega\otimes\psi_1},\\
&V\pi_{\omega\otimes\psi_1\otimes\psi_2}(B_2)
=\rho(B_2)V,\quad\text{for all}\; B_2\in\caB_2,\\
&\rho_1(X)\otimes\unit_{\caH_{\psi_2}}
=\Ad V^* \rho(X),\quad X\in \caA\otimes\caB_1,\\
&V^*V=VV^*=\rho(\unit),\\
&V_{\ld_0}^* V_{\ld_0}=V_{\ld_0} V_{\ld_0}^*=\rho_1(\unit).
\end{split}
\end{align}
For this $\gamma$,
$\fol{\omega\otimes\psi_1\otimes\psi_2}{\omega\otimes\psi_1}(\gamma)=\hat\gamma$,
with
\begin{align}
\begin{split}
\hat\gamma\lmk X\otimes B_2\rmk
=\gamma(X)\otimes\pi_{\psi_2}(B_2),\quad
X\in \caA\otimes\caB_1,\; B_2\in \caB_2.
\end{split}
\end{align}
(Recall the proof of Lemma \ref{washi}).

Set $W:=V\lmk V_{\ld_0}\otimes\unit_{\caH_{\psi_2}}\rmk\in \pi_{\omega\otimes\psi_1\otimes\psi_2}
\lmk \caA\otimes\caB_1\otimes\caB_2\rmk''$.
Then we have
\begin{align}
\begin{split}
&W \hat\gamma\lmk X\otimes B_2\rmk
=V\lmk V_{\ld_0}\otimes\unit_{\caH_{\psi_2}}\rmk \lmk
\gamma(X)\otimes\pi_{\psi_2}(B_2)\rmk\\
&=
V\lmk V_{\ld_0} V_{\ld_0}^* \rho_1(X) V_{\ld_0} \otimes \pi_{\psi_2}(B_2)\rmk
=V\lmk\rho_1(X)\otimes\pi_{\psi_2}(B_2)\rmk\lmk V_{\ld_0}\otimes\unit_{\psi_2}\rmk\\
&=V V^* \rho(X) V \lmk\unit_{\omega\otimes\psi_1}\otimes \pi_{\psi_2}(B_2)\rmk
\lmk V_{\ld_0}\otimes\unit_{\psi_2}\rmk\\
&=V V^* \rho(X \otimes B_2) V
\lmk V_{\ld_0}\otimes\unit_{\psi_2}\rmk
= \rho(X \otimes B_2) W,
\end{split}
\end{align}
for all $X\in \caA\otimes\caB_1$, $B_2\in\caB_2$.
Furthermore, we have
\begin{align}
\begin{split}
&W^*W=\lmk V\lmk V_{\ld_0}\otimes\unit_{\caH_{\psi_2}}\rmk\rmk^* V\lmk V_{\ld_0}\otimes\unit_{\caH_{\psi_2}}\rmk
=\gamma(\unit)\otimes \unit_{\caH_{\psi_2}}=\hat\gamma(\unit)
=\id_{\hat\gamma},\\
&WW^*=V(V_{\ld_0}V_{\ld_0}^*\otimes \unit_{\caH_{\psi_2}})V^*=VV^*=\rho(\unit)
=\id_{\rho}.
\end{split}
\end{align}
Therefore, $W$ is an isomorphism in $\cols{\omega\otimes\psi_1\otimes\psi_2}$
from $\hat\gamma$ to $\rho$.
This completes the proof of the essential surjectivity.
\end{proofof}
\begin{proofof}[Corollary \ref{yofuke}]
Let $\tau : \caA\otimes \caB_1\otimes\caB_2\to \caA\otimes\caB_2\otimes\caB_1$
be the flip $*$-isomorphism
\begin{align}
\tau(A\otimes B_1\otimes B_2)= A\otimes B_2\otimes B_1,\quad A\in\caA, \quad B_1\in\caB_1,\quad B_2\in \caB_2.
\end{align}
By the uniqueness of GNS representations,
there exists a unitary $U: \caH_\omega\otimes\caH_{\psi_1}\otimes\caH_{\psi_2}
\to\caH_\omega\otimes\caH_{\psi_2}\otimes\caH_{\psi_1}$
such that
\begin{align}
\begin{split}
\Ad U \lmk \pi_\omega\otimes\pi_{\psi_1}\otimes\pi_{\psi_2}\rmk
=\lmk \pi_\omega\otimes\pi_{\psi_2}\otimes\pi_{\psi_1}\rmk \tau
\end{split}
\end{align}
It is then easy to see that 
\begin{align}
\begin{split}
&H(\rho):=\Ad U \rho\tau^{-1},\quad H(R):=U RU^*,\\
&\quad \rho,\sigma \in \Obj \tcols{\omega\otimes\psi_1\otimes\psi_2},\\
&\quad R\in \Mor_{\tcols{\omega\otimes\psi_1\otimes\psi_2}}(\rho,\sigma)
\end{split}
\end{align}
defines a equivalence $H$ between the braided $C^*$-tensor functors
$\tcols{\omega\otimes\psi_1\otimes\psi_2}$ and
$\tcols{\omega\otimes\psi_2\otimes\psi_1}$.
Hence, by Proposition \ref{sava}, 
$\cols{\omega\otimes\psi_1\otimes\psi_2}$ and
$\cols{\omega\otimes\psi_2\otimes\psi_1}$
are equivalent as braided $C^*$-tensor categories.
By Theorem \ref{stab},
$\cols{\omega\otimes\psi_1}$ and $\cols{\omega\otimes\psi_1\otimes\psi_2}$
are equivalent. Similarly, 
$\cols{\omega\otimes\psi_2}$ and $\cols{\omega\otimes\psi_2\otimes\psi_1}$
are equivalent.
Hence $\cols{\omega\otimes\psi_1}$ and $\cols{\omega\otimes\psi_2}$ are equivalent.

This completes the proof.
\end{proofof}
The pure state case, Proposition \ref{kapibara}, also follows from Lemma \ref{aiiro}.
\begin{proofof}[Proposition \ref{kapibara}]
By Lemma \ref{washi},
$\fol{\omega\otimes\psi}\omega : \col\to \cols{\omega\otimes\psi}$ is fully faithful.
It remains to show that it is essentially surjective.

Set $\pi:=\pi_\omega\otimes\pi_{\psi}$.
To show the essential surjectivity, let $\rho\in \Obj \cols{\omega\otimes\psi}$.
We would like to find a $\rho_\omega\in\Obj \col$ such that
$\fol{\omega\otimes \psi}\omega(\rho_\omega)$ and $\rho$
are isomorphic in $\cols{\omega\otimes \psi}$.
Because of Lemma \ref{shachi} with $p=\rho(\unit)$, it suffices to consider the case that $\rho(\unit)$ belongs to the center
$Z\lmk \pi\lmk
\caA\otimes \caB\rmk''\rmk=\bbC\unit$, i.e., 
$\rho(\unit)=\unit$.

Now we apply Lemma \ref{orange}
with $\caA$, $\caA_1$
$\caB_1$,
$\caB_2$, $\caH_0$, $\pi_0$,
$\caK_1$, $\pi_1$,
$\caK_2$, $\pi_2$,
$U$, $\rho$, $z_\rho$
replaced by
$\caA$, 
$\caA_{\ld_0}$,
$\caB_{\ld_0}$,
$\caB_{\ld_0^c}$,
$\caH_\omega$,
$\pi_\omega$,
$\caH_{\psi\vert_{\caB_{\ld_0}}}$,
$\pi_{\psi\vert_{\caB_{\ld_0}}}$,
$\caH_{\psi\vert_{\caB_{\ld_0^c}}}$,
$\pi_{\psi\vert_{\caB_{\ld_0^c}}}$,
$V_{\rho\ld_0^c}$, $\rho$, 
$\unit_{\caH_\omega}$.
Then 
we have
\begin{align}
\begin{split}
&\Ad \lmk V_{\rho\ld_0^c}\rmk\pi\lmk A_1\otimes B_1\rmk=\rho(A_1\otimes B_1),\quad
A_1\in \caA_{\ld_0},\quad
B_1\in \caB_{\ld_0},\\
& V_{\rho\ld_0^c}^* V_{\rho\ld_0^c}
\in \pi\lmk\lmk\caA\otimes\caB\rmk_{\ld_0}\rmk',
\\
&\rho\lmk B_2\rmk=\unit_{\caH_\omega}\otimes\pi_\psi(B_2),\quad
B_2\in \caB_{\ld_0^c}
\end{split}
\end{align}
and the conditions of Lemma \ref{orange} are satisfied.
Applying Lemma \ref{orange} with $p=\unit$,
we obtain a partial isometry
\begin{align}
V\in\pi\lmk\caA\otimes\caB_{\ld_0}\rmk''
=\pi_\omega\lmk \caA\rmk''\bar\otimes \pi_{\psi\vert_{\caB_{\ld_0}}}(\caB_{\ld_0})''
\bar \otimes \bbC\unit_{\caH_{\psi\vert_{\caB_{\ld_0^c}}}}
\end{align}
satisfying
\begin{align}\label{automorphic}
\begin{split}
&V \pi\lmk A_1\otimes B_1\rmk=\rho\lmk A_1\otimes B_1 \rmk V,\quad
A_1\in\caA_{\ld_0},\quad B_1\in\caB_{\ld_0},\\
&F\otimes\unit_{\caH_{\psi}}:=V^*V\in 
\lmk \pi_\omega(\caA_{\ld_0})'\cap \pi_\omega\lmk\caA\rmk''\rmk
\bar\otimes\unit_{\caH_\psi}
=\pi_\omega(\caA_{\ld_0})'\bar\otimes\unit_{\caH_\psi}.
\end{split}
\end{align}
In the last equality, we used that $\omega$ is pure.
We would like to modify $V$ to a unitary.

In order to relate $V^*V$ with $\unit$,
note that  the projection $F\in \pi_\omega\lmk\caA\rmk''\cap  \pi_\omega\lmk\caA_{\ld_0}\rmk'\subset 
\caB(\caH_\omega)$ defined
in (\ref{automorphic}) is infinite in $\caB(\caH_\omega)$,
 because $\pi_\omega(\caA_{\ld_0})''$ is an infinite factor.
%We claim that the projection $F\in \pi_\omega\lmk\caA\rmk''=
%\caB(\caH_\omega)$ defined
%above is infinite in $\caB(\caH_\omega)$.
%In fact, because $\pi_\omega(\caA_{\ld_0})''$ is infinite,
%there exists an isometry $\tilde v\in \pi_\omega(\caA_{\ld_0})''$
%with $\tilde v\tilde v^*\neq\unit$.
%Because $F\in \pi_\omega\lmk\caA_{\ld_0}\rmk'$,
%we have $(F\tilde v)^* F\tilde v=F$.
%Note that $\pi_\omega(\caA_{\ld_0})''$ is a factor because $\omega$ is pure, hence
%the $\sigma$-weak homomorphism $\pi_\omega(\caA_{\ld_0})''\ni x\mapsto xF$
%is faithful.
%Therefore, we have $F\neq F \tilde v (F\tilde v)^*$.
%Hence $F$ is infinite in $\caB(\caH_\omega)$.
Therefore,  by Corollary 6.3.5 of \cite{KR},
$F$ is equivalent to $\unit_{\caH_\omega}$ in $\caB(\caH_\omega)$, i.e.,
there exists an isometry $v\in \caB(\caH_\omega)$ with $vv^*=F$.

Next we relate $VV^*$ and $\unit$.
Because $\omega,\psi$ are pure, $\pi(\caA_{\ld_0})''=\pi_\omega(\caA_{\ld_0})''\bar\otimes \bbC\unit_{\psi}$,
$\pi(\caB_{\ld_0})''=\bbC \unit_{\omega}\bar \otimes \pi_{\psi}(\caB_{\ld_0})''$, and
$\pi(\caA_{\ld_0}\otimes \caB_{\ld_0})''
=\pi_\omega(\caA_{\ld_0})''\bar\otimes \pi_{\psi}(\caB_{\ld_0})''$
are factors.
Note that by $\Ad V_{\rho\ld_0^c}$,
$\rho(\caA_{\ld_0})''$ (resp. $\rho(\caB_{\ld_0})''$, 
$\rho(\caA_{\ld_0}\otimes
\caB_{\ld_0})''$) and the factor
$\pi(\caA_{\ld_0})''$ (resp. $\pi(\caB_{\ld_0})''$,
$\pi(\caA_{\ld_0}\otimes
\caB_{\ld_0})''$) are isomorphic.
By the assumption, they are all infinite factors.
Note that $\rho(\caB_{\ld_0})''\subset \rho\lmk\caB_{\ld_0^c}\rmk'=
\pi \lmk\caB_{\ld_0^c}\rmk'$ is of the form
$\caM_1\bar\otimes\bbC\unit_{\caH_{\psi\vert_{\caB_{\ld_0^c}}}}:=\rho(\caB_{\ld_0})''\subset
\pi\lmk\caA\otimes\caB_{\ld_0}\rmk''$
with $\caM_1\subset \caB(\caH_\omega\otimes\caH_{\psi\vert_{\caB_{\ld_0}}})$.
Because $\rho(\caB_{\ld_0})''$ is a factor, $\caM_1$ is also a factor.
Therefore,
$\caN_1:=\pi\lmk\caA\otimes\caB_{\ld_0}\rmk''\cap \rho(\caB_{\ld_0})'
=\caM_1'\bar\otimes\bbC\unit_{\caH_{\psi\vert_{\caB_{\ld_0^c}}}}$
is a factor.
Note that
\begin{align}
\begin{split}
G:=VV^*\in \pi\lmk\caA\otimes\caB_{\ld_0}\rmk''\cap\rho\lmk
\lmk \caA\otimes\caB\rmk_{\ld_0}\rmk'
\subset 
\pi\lmk\caA\otimes\caB_{\ld_0}\rmk''\cap\rho\lmk
 \caB_{\ld_0}\rmk'=\caN_1.
\end{split}
\end{align}
We claim that $G$ is infinite in $\caN_1$.
In fact, as we noted above, $\rho(\caA_{\ld_0})''\lmk \simeq \pi\lmk \caA_{\ld_0}\rmk''\rmk$
is an infinite factor.
Because of $$\rho(\caA_{\ld_0})''\subset \rho(\caB_{\ld_0^c})'=\pi(\caB_{\ld_0^c})'=\pi(\caA\otimes\caB_{\ld_0})'',$$
$\rho(\caA_{\ld_0})''$ is
 a subalgebra of $\caN_1$.
 Since $\rho(\caA_{\ld_0})''$ is an infinite von Neumann algebra,
there exists an isometry $v'\in \rho(\caA_{\ld_0})''\subset \caN_1$ satisfying $v'v'^*\neq \unit$.
Note that $G\in \rho(\caA_{\ld_0})'$ gives a  $\sigma$-weak continuous homomorphism
$\rho(\caA_{\ld_0})''\ni x\mapsto xG$.
Because $\rho(\caA_{\ld_0})''$ is a factor, this is faithful.
Therefore, $G=(Gv')^*(Gv')$ and $(Gv')(Gv')^*$ are different, and
because $Gv'\in \caN_1$, it means that $G$ is infinite in $\caN_1$.
The same $v'$ shows that $\unit$ is infinite in $\caN_1$.
Therefore, by Corollary 6.3.5 of \cite{KR},
$G$ is equivalent to $\unit$ in $\caN_1$, i.e.,
there is an isometry $w$ in $\caN_1$ such that $ww^*=G$.

Now using  $v\in \caB(\caH_\omega)$, $w\in \caN_1$  as above,
set
\begin{align}
W:=w^* V(v\otimes\unit_\psi)\in
\pi\lmk\caA\otimes\caB_{\ld_0}\rmk''.
\end{align}
Then for any $B_1\in \caB_{\ld_0}$, we have
\begin{align}
\begin{split}
&W\pi(B_1)
=w^* V(v\otimes\unit)\lmk \unit_{\caH_\omega}\otimes \pi_\psi(B_1)\rmk
=w^* V\lmk \unit_{\caH_\omega}\otimes \pi_\psi(B_1)\rmk
(v\otimes\unit)\\
&=w^* V \pi(B_1) (v\otimes\unit)
=w^* \rho(B_1) V(v\otimes\unit)
= \rho(B_1) w^* V(v\otimes\unit)
=\rho(B_1) W,
\end{split}
\end{align}
using $w\in \caN_1\subset \rho(\caB_{\ld_0})'$.
Because $W$ belongs to $\pi\lmk\caA\otimes\caB_{\ld_0}\rmk''$,
it commutes with $\rho(B_2)=\pi(B_2)$ for any $B_2\in \caB_{\ld_0^c}$,
hence we obtain
\begin{align}\label{moss}
\begin{split}
W\pi (B)
=\rho(B) W,\quad B\in \caB.
\end{split}
\end{align}
Furthermore, we have
\begin{align}
\begin{split}
&W^*W=(v^*\otimes\unit) V^* ww^* V(v\otimes\unit)
=(v^*\otimes\unit) V^* G V(v\otimes\unit)
=(v^*\otimes\unit) V^*V(v\otimes\unit)
=v^* Fv \otimes\unit=\unit,\\
& WW^*=w^* V(v\otimes\unit)(v^*\otimes\unit) V^* w
=w^* V(F\otimes\unit) V^* w
=w^* VV^* w
=w^* G w=\unit,
\end{split}
\end{align}
hence $W$ is a unitary in $ \pi\lmk\caA\otimes\caB_{\ld_0}\rmk''$.

Now, because of the unitarity of $W$ and (\ref{moss}) we have
\begin{align}
\begin{split}
\Ad W^*\lmk \rho\lmk\caA\rmk\rmk
\subset \Ad W^* \lmk \rho(\caB)'\rmk
=\pi(\caB)'=\pi_\omega(\caA)''\otimes\bbC\unit_{\caH_\psi}.
\end{split}
\end{align}
Therefore, there exists a linear map $\rho_1:\caA\to\caB(\caH_\omega)$
such that
\begin{align}
\begin{split}
\Ad W^*\rho(A)=\rho_1(A)\otimes\unit_{\caH_\psi},\quad
A\in \caA.
\end{split}
\end{align}
Because $W$ is a unitary, this $\rho_1$ is a representation
with $\rho_1(\unit)=\unit_{\caH_\omega}$.
Using this $\rho_1$ and (\ref{moss}), we obtain
\begin{align}
\begin{split}
\rho_1\otimes\pi_\psi=\Ad W^*\circ\rho.
\end{split}
\end{align}

Next, we claim for any cone $\ld$, there exists a unitary $U_\ld$
on $\caH_\omega$ such that
\begin{align}\label{koke}
\begin{split}
\Ad U_\ld\pi_\omega(A)
=\rho_1(A),\quad A\in \caA_{\ld^c}.
\end{split}
\end{align}
In fact, for any $
A\in \caA_{\ld^c}$,
we have
\begin{align}
\begin{split}
\rho_1(A)\otimes\unit
=\Ad W^*\rho(A\otimes\unit)
=\Ad W^* V_{\rho\ld}\pi(A\otimes\unit)
=\Ad W^* V_{\rho\ld}\lmk \pi_\omega(A)\otimes\unit\rmk.
\end{split}
\end{align}
From this, there is a $*$-isomorphism
$\tau: \pi_\omega(\caA_{\ld^c})'' \to \rho_1(\caA_{\ld^c})''$
such that $\tau\lmk\pi_\omega(A)\rmk=\rho_1(A)$ for all $A\in \caA_{\ld^c}$.
In particular, $\rho_1(\caA_{\ld^c})''$ is an infinite factor.
Similarly, $\rho_1(\caA_{\ld})''$ is an infinite factor.
Note that the commutant
$\pi_\omega(\caA_{\ld^c})'$ of $\pi_\omega(\caA_{\ld^c})''$ 
 is a factor including
an infinite factor $\pi_\omega(\caA_{\ld})''$.
Therefore, $\pi_\omega(\caA_{\ld^c})'$ is an infinite factor.
Similarly, the commutant $\rho_1(\caA_{\ld^c})'$
of $\rho_1(\caA_{\ld^c})''$ is also a factor
including an infinite factor $\rho_1(\caA_\ld)''$.
Therefore, $\rho_1(\caA_{\ld^c})'$ is an infinite factor.
From these, using Corollary 8.12 of \cite{struatilua2019lectures}, the  $*$-isomorphism
$\tau: \pi_\omega(\caA_{\ld^c})'' \to \rho_1(\caA_{\ld^c})''$
is spatial. Namely, there exists a unitary
$U_\ld$ on $\caH_\omega$ such that
\begin{align}
\begin{split}
\Ad U_\ld(x)=\tau(x),\quad x\in  \pi_\omega(\caA_{\ld^c})''.
\end{split}
\end{align}
In particular (\ref{koke}) holds.

Now we claim
\begin{align}
\rho_\omega:=\Ad U_{\ld_0}^*\rho_1 \in\ool.
\end{align}
Because $U_{\ld_0}$ is a unitary, this gives a representation of $\caA$
on $\caH_\omega$.
For any cone $\ld$, we have
\begin{align}
\begin{split}
\rho_\omega(A)=\Ad U_{\ld_0}^*\rho_1(A)
=\Ad U_{\ld_0}^* U_\ld \pi_\omega(A),\quad
A\in \caA_{\ld^c},
\end{split}
\end{align}
with $U_{\ld_0}^* U_\ld\in \caB(\caH_\omega)=\pi_\omega(\caA)''$
a unitary.
Hence, $U_{\ld_0}^* U_\ld$ belongs to $\caV_{\rho_\omega\ld}$,
and $\rho_\omega$ belongs to $\ool$.

Finally, we claim that $\hat\rho_\omega:=\fol{\omega\otimes\psi}{\omega}(\rho_\omega)$
is isomorphic to $\rho$ in $\cols{\omega\otimes\psi}$.
In fact, recalling that
the isomorphism of Lemma \ref{hina} is
$\Theta^{-1}(x)=x\otimes\unit$, $x\in\pi_\omega(\caA)''$ in this setting,
 we have
\begin{align}
\begin{split}
\hat\rho_\omega\lmk A\otimes B\rmk
=\Theta^{-1}\lmk\rho_\omega(A)\rmk\pi(B)
=\rho_\omega(A)\otimes \pi_\psi(B),
\end{split}
\end{align}
for any $A\in\caA$, $B\in \caB$.
Set
\begin{align}
\begin{split}
U:=W(U_{\ld_0}\otimes\unit_{\caH_\psi})
\in \caU\lmk \pi(\caA\otimes\caB)''\rmk
=\caU\lmk \caB\lmk\caH_\omega\otimes\caH_\psi\rmk\rmk.
\end{split}
\end{align}
We have
$\rho(\unit )U\hat\rho_\omega(\unit)=U$
and
\begin{align}
\begin{split}
&U\hat\rho_\omega(A\otimes B)
=U\lmk\rho_\omega(A)\otimes\pi_\psi(B)\rmk
=W(U_{\ld_0}\otimes\unit_{\caH_\psi}) 
\lmk U_{\ld_0}^* \rho_1(A) U_{\ld_0}\otimes\pi_\psi(B)\rmk\\
&=W \lmk \rho_1(A)\otimes\pi_\psi(B)\rmk \lmk U_{\ld_0}\otimes\unit_{\caH_\psi} \rmk
=\rho(A\otimes B) W\lmk U_{\ld_0}\otimes\unit_{\caH_\psi} \rmk
=\rho(A\otimes B) U,
\end{split}
\end{align}
for any $A\in\caA$ and $B\in \caB$.
This proves that the unitary $U$ belongs to
$\Mor_{\cols{\omega\otimes\psi}}(\hat\rho_\omega,\rho)$.
This completes the proof of the essential surjectivity.
\end{proofof}

\section{Proof of Theorem \ref{main}} \label{mainprf}
In this section we prove our main theorem, Theorem \ref{main}.

Let us first recall the definition of approximately-factorizable automorphisms.
\begin{defn}\label{qfdef}
Let $\alpha$ be an automorphism of a two-dimensional quantum spin system $\caA$.
We say that $\alpha$ is {approximately-factorizable}
if the following condition holds.
\begin{description}
\item[(i)]
For any cone $\Lambda$ and $\delta>0$,
there are automorphisms $\alpha_{\Lambda}, \tilde \alpha_{\Lambda}\in\Aut\lmk\caA_{\Lambda}\rmk$,
$\alpha_{\Lambda^c}, \tilde \alpha_{\Lambda^c}\in \Aut\lmk\caA_{\Lambda^c}\rmk$
and $\Xi_{\Lambda, \delta}, \tilde \Xi_{\Lambda, \delta}\in \Aut\lmk\caA_{\Lambda_{\delta}\cap (\Lambda^c)_\delta}\rmk$
and unitaries $v_{\Lambda\delta},\tilde v_{\Lambda\delta}\in\at$
such that 
\begin{align}
\begin{split}
&\alpha=\Ad\lmk v_{\Lambda\delta}\rmk\circ
\Xi_{\Lambda, \delta}\circ \lmk\alpha_{\Lambda}\otimes \alpha_{\Lambda^c}\rmk,\\
&\alpha^{-1}=\Ad\lmk \tilde v_{\Lambda\delta}\rmk\circ
\tilde \Xi_{\Lambda, \delta}\circ \lmk\tilde \alpha_{\Lambda}\otimes \tilde \alpha_{\Lambda^c}\rmk.
\end{split}
\end{align}
\item[(ii)]
For each $\delta,\delta'>0,\varphi\in (0,2\pi)$, there exists a decreasing
function $g_{\varphi,\delta,\delta'}(t)$ on $\bbR_{\ge 0}$
with $\lim_{t\to\infty}g_{\varphi,\delta,\delta'}(t)=0$.
For any cone $\Lambda$ with $\varphi=|\arg\Lambda|$, there are unitaries
 $v_{\Lambda,\delta,\delta',t}', \tilde v_{\Lambda,\delta,\delta',t}'\in \caA_{\Lambda_{\delta+\delta'}-t\bm e_{\Lambda}}$
 satisfying
\begin{align}\label{uappro}
\lV
v_{\Lambda,\delta}-v_{\Lambda,\delta,\delta',t}'
\rV,\lV
\tilde v_{\Lambda,\delta}-\tilde v_{\Lambda,\delta,\delta',t}'
\rV
\le g_{\varphi,\delta,\delta'}(t),
\end{align}
for unitaries $v_{\Lambda\delta},\tilde v_{\Lambda\delta}$ in (i).
\end{description}
\end{defn} 
We note that approximately factorizable automorphisms preserve
the approximate Haag duality and the properly infiniteness of cone algebras.
\begin{lem}
Let $\omega$ be a state on a two-dimensional quantum spin system $\caA$
with properly infinite cone algebras, satisfying the approximate Haag duality.
Let $\alpha$ be an approximately factorizable automorphism on $\caA$.
Then $\omega\alpha$
has properly infinite cone algebras and satisfies the approximate Haag duality.
\end{lem}
\begin{proof}
By the same proof as \cite{MTC}, we can show that
$\omega\alpha$ satisfies the approximate Haag duality.

Let $(\caH,\pi)$ be a GNS representation of $\omega$. For any cone $\ld$, with $\delta>0$ small enough,
we have
\begin{align}
\begin{split}
\pi\alpha(\caA_\ld)
=\Ad\lmk \pi( v_{\ld\delta} )\rmk \lmk \pi\Xi_{\ld\delta}(\caA_\ld)\rmk
\supset 
\Ad\lmk \pi( v_{\ld\delta} )\rmk \lmk
\pi (\caA_{\ld_{-\delta}})\rmk,
\end{split}
\end{align}
and the von Neumann algebra $\pi\alpha(\caA_\ld)''$ includes
a properly infinite von Neumann algebra $\Ad\lmk \pi( v_{\ld\delta} )\rmk \lmk
\pi (\caA_{\ld_{-\delta}})''\rmk$.
Hence $\pi\alpha(\caA_\ld)''$ is properly infinite,
and $\omega\alpha$  has  properly infinite cone algebras.\
\end{proof}

Now we prove Theorem \ref{main}.
\begin{proofof}[Theorem \ref{main}]
We denote by $\alpha_{13}$ the automorphism $\alpha$ acting on the first and third components of
$\caA\otimes\caB_2\otimes\caB_1$.
By the same proof as \cite{MTC}, 
the braided $C^*$-tensor categories $\tcols{\omega_1\otimes\psi_2\otimes\psi_1}$ and
$\tcols{(\omega_1\otimes\psi_2\otimes\psi_1)\alpha_{13}}$ are equivalent.
By Proposition \ref{sava}, this means
$\cols{\omega_1\otimes\psi_2\otimes\psi_1}$ and
$\cols{(\omega_1\otimes\psi_2\otimes\psi_1)\alpha_{13}}$ are equivalent.
By Theorem \ref{stab}, 
$\cols{\omega_1\otimes\psi_2}$ and  $\cols{\omega_1\otimes\psi_2\otimes\psi_1}$ 
are equivalent.
By Theorem \ref{cofe}, there is a faithful braided tensor functor 
from $\cols{\omega_2\otimes\psi_2}=\cols{\lmk \omega_1\otimes\psi_2\otimes\psi_1\rmk\alpha_{13}\vert_{\caA\otimes\caB_2}}$
to $\cols{\lmk \omega_1\otimes\psi_2\otimes\psi_1\rmk\alpha_{13}}$,
hence to $\cols{\omega_1\otimes\psi_2}$.
This proves the Theorem.

\end{proofof}
{\bf Acknowledgment.}\\
The author is grateful to David P\'erez-Garc\' ia and Angelo Lucia for the fruitful discussion.

\appendix
\section{Notation}\label{notation}
The unit of a $C^*$-algebra $\mathfrak A$ will be denoted by $\unit_{\mathfrak A}$,
although we will frequently omit the subscript and write $\unit$.
The set of all unitaries in $\mathfrak A$ is denoted by $\caU(\mathfrak A)$.
For a state $\omega$ on a $C^*$-algebra $\mathfrak A$, when we write
$(\caH_\omega,\pi_\omega,\Omega_\omega)$, it means a GNS triple
of $\omega$.
For two UHF algebras $\mathfrak A$, $\mathfrak B$, we occasionally denote the sub-algebra
$\mathfrak A\otimes\unit_{\mathfrak B}$ (resp. $\unit_{\mathfrak A}\otimes\mathfrak B$) of 
$\mathfrak A\otimes\mathfrak B$ by
$\mathfrak A$ (resp. $\mathfrak B$).
With the same spirit, we frequently denote $A\otimes\unit_{\mathfrak B},
\unit_{\mathfrak A}\otimes  B\in \mathfrak A\otimes\mathfrak B$ by $A$, $B$ respectively.
We denote by $\Aut(\mathfrak A)$ the automorphism group of $\mathfrak A$.

For a von Neumann algebra $\caM$, we denote its center by $Z(\caM)$.
We denote by $Z_{\caM}(x)$ the central carrier of $x\in\caM$ in $\caM$. 
For von Neumann algebras $\caM$, $\caN$, $\caM\bar \otimes\caN$ is the von Neumann algebra
tensor product of them.

 For a Hilbert space $\caH$, $\caB(\caH)$ denotes the algebra of all bounded operators
on $\caH$. The identity map on a Hilbert space $\caH$ is denoted by $\unit_{\caH}$, but the subscript is occationally 
omitted. For a state $\omega$, $\unit_{\caH_\omega}$
is also denoted by $\unit_\omega$ for simplicity.
For operators $V,x$ on a Hilbert space $\caH$, we set $\Ad V(x):=V x V^*$.

\noindent{Funding and/or Conflicts of interests/Competing interests This work was supported by JSPS KAKENHI Grant Number 19K03534 and 22H01127.
It was also supported by JST CREST Grant Number JPMJCR19T2.
The author has no competing interests to declare that are relevant to the content of this article.}\\\\
\noindent{DATA AVAILABILITY
The data that support the findings of this study are available within the article.}

\bibliographystyle{alpha}
\bibliography{MTO}

\end{document}